\documentclass[preprint,spthm,nohyperref]{iacrtrans}
\usepackage[utf8]{inputenc}
\usepackage{amsmath,amsfonts,amssymb}
\usepackage{mathtools}
\usepackage{xspace}
\usepackage{graphicx}
\usepackage[ruled,vlined,linesnumbered]{algorithm2e}
\DontPrintSemicolon

\usepackage{listings}

\usepackage{booktabs}
\usepackage{multirow}
\usepackage{physics}
\usepackage{stmaryrd}
\usepackage{color}
\usepackage{nicefrac}
\usepackage{thm-restate}
\usepackage{calrsfs}
\usepackage[binary-units]{siunitx}
\usepackage{mdframed}
\usepackage{units}

\usepackage{tikz}
\usetikzlibrary{calc}
\usetikzlibrary{shapes.geometric}
\usetikzlibrary{arrows}
\usetikzlibrary{automata}
\usetikzlibrary{positioning}

\newcommand{\Wave}{\textsc{Wave}\xspace}
\newcommand{\Wavelet}{\textsc{Wavelet}\xspace}

\newcommand{\Supertubos}{\emph{Supertubos}\xspace}

\newcommand{\FF}{\ensuremath{\mathbb{F}}}
\newcommand{\F}{\ensuremath{\mathbb{F}}}
\newcommand{\ZZ}{\ensuremath{\mathbb{Z}}}
\newcommand{\Fq}{\ensuremath{\F_q}}

\renewcommand{\vec}[1]{\mathbf{#1}}

\makeatletter
\newcommand*{\transp}{{\mathpalette\@transpose{}}}
\newcommand*{\@transpose}[2]{\raisebox{\depth}{$\m@th#1\intercal$}}
\makeatother
\newcommand{\transpose}[1]{{#1}^{\transp}}

\newcommand*{\eqdef}{:=}

\newcommand{\Hpub}{\vec{H}_{\textup{pk}}}
\newcommand{\Hsec}{\vec{H}_{\textup{sk}}}
\newcommand{\sk}{\ensuremath{\mathtt{sk}}}
\newcommand{\pk}{\ensuremath{\mathtt{pk}}}
\newcommand{\Diag}{\mathbf{Diag}}
\newcommand{\hsp}{\odot}

\newcommand{\Unif}{\xleftarrow{\$}}
\newcommand{\pisec}{\pi_{\textup{sk}}}

\newcommand\sample{\mathrel{\overset{\vbox to.7ex{\hbox{\fontsize{6}{0}\selectfont\vspace{.5ex} \$}}}{\leftarrow}}}

\newcommand{\Sym}[1]{\mathcal{S}_{#1}}

\SetKwProg{Function}{Function}{}{}
\SetKwInput{KwParameters}{Parameters}
\SetKwInput{KwPrecomputation}{Precomputation}
\SetKwFunction{FreeSetU}{FreeSetU}
\SetKwFunction{FreeSetV}{FreeSetV}
\SetKwFunction{ParityCheckUV}{ParityCheckUV}
\SetKwFunction{SubmatrixRank}{SubmatrixRank}
\SetKwFunction{RandPerm}{RandPerm}
\SetKwFunction{RandPermSet}{RandPermSet}
\SetKwFunction{SwapPerm}{SwapPerm}
\SetKwFunction{GaussElim}{GaussElim}
\SetKwFunction{ElimSingle}{ElimSingle}
\SetKwFunction{AcceptSampleU}{AcceptSampleU}
\SetKwFunction{AcceptSampleV}{AcceptSampleV}
\SetKwFunction{KeyGen}{KeyGen}
\SetKwFunction{Sign}{Sign}
\SetKwFunction{Decode}{Decode}
\SetKwFunction{SampleU}{SampleU}
\SetKwFunction{SampleV}{SampleV}
\SetKwFunction{Verify}{Verify}
\SetKwFunction{TVerify}{TVerify}
\SetKwFunction{STVerify}{STVerify}
\SetKwFunction{Convert}{Convert}
\SetKwFunction{Hash}{Hash}
\SetKwFunction{BinaryHash}{BinaryHash}
\SetKwFunction{Ternarize}{Ternarize}
\SetKwFunction{Expand}{Expand}
\SetKwFunction{Supp}{Support}
\SetKwFunction{Compress}{Compress}
\SetKwFunction{Decompress}{Decompress}
\SetKwFunction{MREncode}{MREncode}
\SetKwFunction{MRDecode}{MRDecode}
\SetKwFunction{NextPair}{NextPair}
\SetKwFunction{Positions}{Positions}
\SetKwFunction{FetchRow}{FetchRow}
\SetKwFunction{BinaryWeight}{BinaryWeight}
\SetKwFunction{Normalize}{NegateRow}
\SetKwFunction{SwapRows}{SwapRows}
\SetKwFunction{SubtractRow}{SubtractRow}
\SetKwFunction{SumRow}{SumRow}
\SetKwFunction{Compression}{Compression}
\SetKwFunction{GetPos}{GetPos}
\SetKwFunction{Length}{CodeLength}
\SetKwFunction{code}{Code}
\SetKwFunction{Decompression}{Decompression}
\SetKwFunction{Dec}{Decode}
\SetKwFunction{Find}{Find}
\SetKw{In}{in}
\SetKw{Where}{where}
\SetKw{And}{and}
\SetKw{True}{True}
\SetKw{False}{False}

\newcommand{\DuoHuff}{2TMR\xspace}
\newcommand{\TriHuff}{3TMR\xspace}

\lstdefinestyle{customc}{
  belowcaptionskip=1\baselineskip,
  breaklines=true,
  frame=L,
  xleftmargin=\parindent,
  language=C,
  showstringspaces=false,
  basicstyle=\footnotesize\ttfamily,
  keywordstyle=\bfseries\color{green!40!black},
  commentstyle=\itshape\color{purple!40!black},
  identifierstyle=\color{blue},
  stringstyle=\color{orange},
}

\begin{document}
\title[\Wavelet: Code-based postquantum signatures for microcontrollers]{\Wavelet: Code-based postquantum signatures with fast
verification on microcontrollers}

\author{Gustavo Banegas\inst{1} \and Thomas Debris-Alazard\inst{1} \and Milena Nedeljković\inst{2} \and Benjamin Smith\inst{1}}
\institute{Inria and Laboratoire d’Informatique de l’École polytechnique, Institut Polytechnique de Paris, Palaiseau, France, \email{gustavo@cryptme.in, thomas.debris@inria.fr, smith@lix.polytechnique.fr} \and
           École polytechnique, Institut Polytechnique de Paris, Palaiseau, France \email{milena.nedeljkovic@polytechnique.edu}}
\authorrunning{Banegas, Debris-Alazard, Nedeljković, Smith}

\maketitle              


\begin{abstract}
    This work presents the first full implementation of \Wave,
    a postquantum code-based signature scheme.
    We define \Wavelet,
    a concrete \Wave scheme
    at the 128-bit classical security level (or NIST postquantum
    security Level~1)
    equipped with a fast verification algorithm targeting embedded devices.
    \Wavelet offers 930-byte signatures,
    with a public key of 3161 kB.
    We include implementation details
    using AVX instructions,
    and on ARM Cortex-M4,
    including a solution to deal with \Wavelet's large public keys,
    which do not fit in the SRAM of a typical embedded device.
    Our verification algorithm is $\approx 4.65 \times$ faster
    then the original,
    and verifies in $1\,087\,538$ cycles
    using AVX instructions,
    or $13\,172$ ticks in an ARM Cortex-M4.
\end{abstract}

\begingroup
\makeatletter
\def\@thefnmark{} \@footnotetext{\relax
Author list in alphabetical order; see
\url{https://www.ams.org/profession/leaders/culture/CultureStatement04.pdf}.
This work was funded in part
by the European Commission through H2020
SPARTA, \url{https://www.sparta.eu/} and
by the French Agence Nationale de la Recherche through ANR JCJC COLA (ANR-21-CE39-0011).
Generic disclaimer:
``Any opinions, findings, and conclusions or recommendations expressed in this material are those
of the author(s) and do not necessarily reflect the views of the funding agencies.''
}
\endgroup

\section{Introduction}
\label{sec:introduction}

Ensuring the long-term security of constrained devices is a critical,
and perennial, problem.
The secure and efficient implementation of public-key cryptographic
algorithms for microcontrollers is vital,
but difficult given the constrained resources available on these devices.
The difficulties are magnified when it comes to implementing
\emph{post-quantum} cryptosystems, which generally have much larger keys and more intensive
computing requirements than their elliptic-curve equivalents.

In this article,
we focus on the problem of implementing post-quantum signature schemes
with efficient verification on microcontrollers.
Our scheme, \Wavelet, is a variant of
\Wave~\cite{2019/Debris-Alazard--Sendrier--Tillich},
a post-quantum trapdoor based on the hardness of problems in coding theory.
\Wave has attractive properties,
including \emph{simple verification} and relatively \emph{short signatures}.
In a sea of broken code-based signature schemes
(including~\cite{Baldi0CRS13}, \cite{PhessoT16},
\cite{SongHMWW20}, and~\cite{AragonBDKPS21}),
\Wave maintains the promise of high security.
\Wave also presents some interesting practical challenges:
for example,
the \Wave designers recommend working over the finite field \(\FF_3\),
which is highly unusual in contemporary cryptographic software.

\Wavelet is a signature scheme based on \Wave over \(\FF_3\),
but designed for efficient verification on constrained devices.
\Wavelet keys and signatures
are easily converted to \Wave keys and signatures over \(\FF_3\),
and vice versa;
in particular,
\Wavelet inherits its security from \Wave over \(\FF_3\).
There are two main benefits to \Wavelet compared with \Wave:
shorter signatures, and faster verification.

\Wave signatures are already remarkably short,
but \Wavelet signatures are even shorter.
We achieve this on two levels:
first, we truncate \Wave signatures from length \(n\) (the code
length) to length \(k\) (the code dimension).
This already represents a substantial improvement:
for our parameters, truncated signatures are one-third shorter.
But these truncated signatures can be compressed further
using minimal-redundancy coding, at minimal computational cost.

The end result is that \Wavelet signatures are the shortest among
code-based signatures, and nearly half the length
of \Wave signatures for the same code parameters.
In practice, when targeting the 128-bit classical security level
(and NIST postquantum Level 1 security),
\Wavelet signatures fit in under one kilobyte.
This is competitive with lattice-based postquantum signatures:
NIST Round-3 candidates Dilithium~\cite{Dilithium} and Falcon~\cite{Falcon}
targeting similar security levels
offer 2420- and 666-byte signatures, respectively
(see Table~\ref{tab:sizes} for a more complete comparison).

\Wavelet's signature verification is also faster---and not just because
the signatures are shorter.
Indeed, a subtle tweak to the public key
allows us to avoid loading or operating on nearly half of its data
during any given signature verification.
This yields an important speedup by drastically reducing the number of
field operations,
but also by reducing the impact of memory latency---an effect that is
further exaggerated in microcontrollers by the fact that the public key
is typically too large for the RAM.

We have put this all into practice
with a free and portable C implementation of \Wavelet for x86\_64,
with verification on x86\_64 (with and without AVX) and on Cortex-M4 platforms,
targeting the 128-bit security level (and NIST Level 1).
This is the first full implementation of a \Wave signature scheme\footnote{%
    A proof-of-concept implementation of the original \Wave trapdoor
    was published with~\cite{2019/Debris-Alazard--Sendrier--Tillich}
    at~\cite{Wave-implem},
    but it cannot be used to sign messages:
    in particular,
    it does not include the ternary hash function
    required to hash messages.
}
on any platform.
We show how to use high-speed bitsliced \(\FF_3\)-vector arithmetic
to speed \Wave systems,
and use external Flash via QSPI
to handle very large public keys (over 3MB).
Using this software,
\Wavelet signatures
can be verified on an Intel\textsuperscript{\tiny\textregistered}
Core\textsuperscript{TM} PC in 450\(\mu\)s,
and on an ARM Cortex-M4 powered board in 402ms.

\begin{table}[htp]
    \caption{Size of private/public keys and signatures
    for postquantum signature schemes.}
\label{tab:sizes}
\begin{center}
\begin{tabular}{r|rrrr}
\toprule
Algorithm & Private Key (\si{\byte}) & Public Key (\si{\byte})& Signature (\si{\byte}) \\
\midrule
 \Wavelet             & $32$ or $11043162$    & $3236327$    & $930$    \\
 \midrule
    Falcon-512~\cite{Falcon}               & $1281$    & $897$   & $666$    \\
    Dilithium2~\cite{Dilithium}            & $2528$    & $1312$  & $2420$    \\
 LMS ~\cite{rfc8554}        & $64$    & $60$    & $4756$    \\
 SPHINCS$^+$-128f~\cite{SPHINCS}     & $64$    & $32$    & $17088$ \\
 GeMSS~\cite{GeMSS}                & $48$    & $14520$  & $417416$ \\
 Picnic3-L1~\cite{Picnic}          & $17$    & $34$     &  $13802$ \\
\bottomrule
\end{tabular}
\end{center}
\end{table}

\paragraph{Organization of the paper}
Section~\ref{sec:high-level} gives
an overview of \Wave,
defines the \Supertubos parameters,
and explains the improvements that lead to our concrete scheme, \Wavelet.
Section~\ref{sec:scheme} gives detailed versions of
the algorithms in \Wavelet. Section~\ref{sec:encodings}
gives useful binary representations for ternary keys,
describes the efficient compression and decompression
of \Wave and \Wavelet signatures,
and calculates key and (compressed) signature sizes.
Section~\ref{sec:implementation}
gives details on the implementation of \Wavelet in software.
Section~\ref{sec:results} shows the experimental results from our
implementation on x86\_64 (with and without AVX) and Cortex-M4 platforms.
We conclude in Section~\ref{sec:conclusion},
giving perspectives on future development.

\paragraph{Notation and conventions}
Vectors and matrices are written in bold. Vectors are in row notation.
Indices start at~\(0\):
if \(\vec{v}\) is a vector in \(\FF_q^N\),
then its components are \((v_0,\ldots,v_{N-1})\);
if \(\vec{H}\) is a matrix in \(\FF_q^{M\times N}\),
then its row vectors are \(\vec{H}_0,\ldots,\vec{H}_{M-1}\) and $\vec{H}_{i,j}$ is its coefficient at position $(i,j)$.
The Hamming weight of \(\vec{v}\)
is \(|\vec{v}| := \#\{ 0 \leq i < N \mid v_{i} \neq 0 \}\).
The concatenation of $\vec{u}$ and $\vec{v}$ is denoted by $\vec{u} \parallel \vec{v}$.

The set of permutations of \(\{0,\ldots,N-1\}\) is denoted
by \(\Sym{N}\).
We represent a permutation \(\pi\) in \(\Sym{N}\)
by the sequence \((\pi(0),\ldots,\pi(N-1))\)
(and we warn the reader that this is \emph{not} cycle notation).
Permutations act on vector coordinates:
\[
    \pi(\vec{v})
    :=
    (v_{\pi(0)},\ldots,v_{\pi({N-1})})\parallel(v_{N},\ldots,v_{N'-1})
    \quad
    \text{for }
    \pi \in \Sym{n}
    \text{ and }
    \vec{v} \in \FF_q^{N'}
    \text{ with }
    N' \ge N
    \,.
\]
If $\vec{H}$ is a matrix in \(\FF_q^{M\times N}\),
then $\pi(\vec{H})$ is the matrix
with row vectors $\pi(\vec{H}_{0}),\ldots, \pi(\vec{H}_{M-1})$.

To ease analysis, the functions and subroutines in our algorithms do \emph{not} modify
their arguments.
In practice, we (strongly) recommend implementing them to modify large vector and
matrix arguments in-place, to save memory.
Finally, if $\mathcal{E}$ is a finite set, then
$x \Unif \mathcal{E}$ means that
$x$ is sampled uniformly at random from $\mathcal{E}$.

\section{
    \Wave and \Wavelet
}
\label{sec:high-level}
We begin with a high-level view of
\Wave,
before explaining the modifications that yield \Wavelet.
This will serve as a high-level tour of our main theoretical results.

\subsection{\Wave}

\Wave is a
full-domain hash (FDH) signature scheme~\cite{BellareR96,Coron02}.
Fix a prime power \(q \not= 2\),
and
let $\vec{H}\in\Fq^{(n-k)\times n}$ be a parity-check matrix for a code
$W = \{ \vec{x}\in\Fq^{n} \mid \vec{x}\transpose{\vec{H}} = \mathbf{0} \}$
of dimension \(k\) and length \(n\)
over \(\FF_q\).
Fix a weight \(w < n\).
If \(W\) is a random code,
then for carefully chosen \(w\),
the function mapping error vectors \(\vec{e}\) in \(\FF_q^n\) of weight \(w\)
to their syndromes \(\vec{s} = \vec{e}\transpose{\vec{H}}\)
is a one-way function:
inverting it corresponds to decoding a random linear code.
But if we suppose that \(W\) is equipped with a trapdoor
allowing us to invert the one-way function,
then we
get the FDH signature scheme of
Figure~\ref{fig:scheme}.

\begin{figure}[ht]
  {\small
  \centering
  \begin{mdframed}\small
   \begin{minipage}{.5\textwidth}
     \underline{Sign$(m, \sk)$:} \\
     \begin{tabular}{ll}
        & $r  \Unif \{0,1\}^{\lambda}$ \\
        & $\vec{s} \gets \texttt{Hash}(m, r)$\\
        & $\vec{e} \gets \texttt{InvAlg}(\vec{s}, \sk)$ \\
        & $\texttt{return}(\vec{e},r)$\\
     \end{tabular}
   \end{minipage}%
   \begin{minipage}{0.5\textwidth}
     \underline{Verify$(m, (\vec{e'}, r), \pk=(\vec{H},w))$:} \\
     \begin{tabular}{ll}
        & $\vec{s} \gets \texttt{Hash}(m, r)$\\
        & $\texttt{If} \text{ }\vec{e'}\transpose{\vec{H}} == \vec{s}$ and $|\vec{e'}| == w$: \\
        & $\qquad\texttt{return}\text{ }\texttt{True}$\\
        & $\texttt{Else:}$\\
        & $\qquad\texttt{return} \text{ } \perp$\\
     \end{tabular}
   \end{minipage}
  \end{mdframed}
  }
  \caption{Code-based signature using FDH paradigm (on a high level).
        The function $\texttt{InvAlg}$ is a ``decoder''
        that inverts the trapdoor,
        correctly finding an error $\vec{e}$ with $|\vec{e}| = w$.
    }
  \label{fig:scheme}
\end{figure}

The \Wave trapdoor is built from two random linear codes \(U\) and \(V\)
of length \(n/2\)
and dimensions \(k_U\) and \(k_V\), respectively,
over \(\FF_q\).
These codes are combined way to form a code \(W\)
of length \(n\) and dimension \(k = k_U + k_V\),
using a construction explained in~\S\ref{sec:keygen}.
The public key is a parity-check matrix \(\Hpub\) in
\(\FF_q^{(n-k)\times n}\) of the code $W$;
the private key consists of \(U\), \(V\),
and data
allowing us to map decoding problems into \(U\) and \(V\).
The parameters \(k_U\), \(k_V\), \(n\), and \(w\)
are carefully balanced,
and the construction of \(W\) from \(U\) and \(V\) carefully designed,
so that the map
\(\vec{e} \mapsto \vec{s} = \vec{e}\transpose{\Hpub}\)
is a trapdoor:
computing a solution \(\vec{e}\) of the correct weight \(w\)
for a given \(\vec{s}\)
is cryptographically hard \emph{unless} we have the secret key,
in which case we can use the decoding algorithm
described in~\S\ref{sec:sign}.

To construct a signature scheme from this trapdoor,
it suffices to let \(\vec{h}\) be the hash of a message \(m\)
with a random salt \(r\);
the \Wave signature is the pair \((\vec{e}, r)\).
To verify the signature,
it suffices to check
that \(\vec{e}\transpose{\Hpub}\) is the hash of \(m\) and \(r\)
and that \(|\vec{e}| = w\).

The theoretical security of \Wave
is beyond the scope of this article.
Details and proofs appear in~\cite{2019/Debris-Alazard--Sendrier--Tillich}.

\subsection{Parameters}
\label{sec:parameters}

The fundamental \Wave parameters
are a classical security parameter \(\lambda\),
a field size \(q \not= 2\),
vector space dimensions \(n\), \(k_U\), and \(k_V\),
a Hamming weight \(w\),
and a parameter \(d\) used in the rejection sampling
in the signing algorithm.
Concrete values for these parameters targeting 128-bit classical
security and NIST Level 1 postquantum security
were specified in~\cite{2019/Debris-Alazard--Sendrier--Tillich},
and are reproduced here in Table \ref{table:genParameters}.
We call this parameter set \Supertubos.\footnote{%
    See \texttt{39.3410294768957, -9.36095797386207}.
}

\begin{table}[ht]
 	\centering
    \begin{tabular}{|c||c|c|c|c|c@{ = }c@{ + }c|c|}
        \hline
            Parameters & $\lambda$ & $q$ & $n$ & $w$ & $k$ & $k_{U}$ & $k_{V}$ & $d$ \\
        \hline
        \hline
            \Supertubos & $128$ & $3$ & $8492$ & $7980$ & $5605$ & $3558$ & $2047$ & $81$ \\
        \hline
 	\end{tabular}
    \caption{General \Wave parameters for $128$ bits of security.\label{table:genParameters}}
\end{table}

Our \Wavelet implementation uses \Supertubos,
and the reader may keep these parameter values in mind
to get a concrete picture of \Wavelet's practical improvements.
The \Wavelet algorithms will work with any \Wave parameter set with \(q = 3\),
though the performance of the signature compression algorithm is sensitive to the ratio \(w/n\approx 0.94\).
Deriving new parameter sets that may be better suited to constrained
evironments (or other applications)
is left for future work.

\subsection{Restricting to \(q = 3\)}
The description of \Wave in~\cite{2019/Debris-Alazard--Sendrier--Tillich}
defines a trapdoor over a general finite field \(\FF_q\) with \(q \not= 2\),
but the security analysis focuses on the case \(q = 3\),
and parameters are only specified for \(q = 3\).
The analysis of \Wave with \(q > 3\) is an open question.

While~\cite{2019/Debris-Alazard--Sendrier--Tillich}
recommends restricting \Wave to \(q = 3\),
\Wavelet explicitly \emph{requires} \(q = 3\).
Forcing \(q = 3\) allows several important practical improvements,
including
\begin{itemize}
    \item
        shorter signatures through more effective compression,
        and more compact keypairs;
    \item
        simple field arithmetic
        and high-speed bitsliced vector operations (see~\S\ref{sec:ternary});
        and
    \item
        interesting algorithmic optimizations,
        especially in signature verification.
\end{itemize}

Using \(q = 3\)
also imposes some interesting practical questions,
since the data to be signed,
and the keys and signatures,
must ultimately be encoded in binary.
The choice of binary encoding for each ternary object is addressed
in~\S\ref{sec:encodings}.

But not all of these problems are a simple matter of data
representation.
For example: \Wave is a hash-and-sign signature scheme,
and it requires a cryptographic hash function \(\{0,1\}^* \to
\FF_3^{n-k}\).
No concrete hash function is specified
in~\cite{2019/Debris-Alazard--Sendrier--Tillich},
and indeed we are not aware of any standard,
high-security hash functions mapping from binary to ternary domains.
Instead, we must build one from a secure hash functions into a binary
domain.
The construction of a hash function into \(\FF_{3}^{n-k}\)
is addressed in~\S\ref{sec:hash}.

\subsection{Truncated signatures and \Wavelet public keys}
\label{sec:Wavelet-signatures}

The \Wave public key is a parity-check matrix \(\Hpub\) for the code \(W\);
but in fact, the basic \Wave key-generation algorithm
always constructs a row-reduced
\(\Hpub = (\vec{I}_{n-k}|\vec{R})\)
where \(\vec{R}\) is a random-looking \((n-k)\times k\) matrix
over \(\FF_3\).
The verification equations are
\begin{equation}
    \label{eq:Wave-verification}
    |\vec{e}| = w
    \qquad
    \text{and}
    \qquad
    \vec{h} = \vec{e}\transpose{\Hpub}
    \qquad
    \text{where}
    \qquad
    \vec{h} = \Hash(r\parallel m)
    \,.
\end{equation}

Now let \(\vec{s}\) in \(\FF_3^{k}\)
consist of the last \(k\) coordinates of \(\vec{e}\);
let
\(\vec{u} := \vec{h} - \vec{s}\transpose{\vec{R}}\)
and \(\vec{e}' := \vec{u}\parallel\vec{s}\).
If \(\vec{e}\transpose{\Hpub} = \vec{h}\),
then \(\vec{e}'\transpose{\Hpub} = \vec{h}\) too
(indeed, typically \(\vec{e} = \vec{e}'\)),
and \(|\vec{e}| = |\vec{e}'| = |\vec{u}| + |\vec{s}|\).
We can therefore transmit \((r,\vec{s})\) as the signature
instead of \((r,\vec{e})\),
thus saving \(n-k\) elements of \(\FF_3\),
and use the verification equation
\begin{equation}
    \label{eq:intermediate-verification}
    |\vec{u}| + |\vec{s}| = w
    \qquad
    \text{where}
    \qquad
    \vec{u} = \vec{h} - \vec{s}\transpose{\vec{R}}
    \qquad
    \text{with}
    \qquad
    \vec{h} = \Hash(r\parallel m)
    \,.
\end{equation}
For \Supertubos, this reduces the uncompressed signature length---and the
work involved in the vector-matrix product---by roughly one-third.

The verifier needs to compute
\(
    \vec{s}\transpose{\vec{R}}
    =
    \sum_{i=0}^{k-1}s_i\transpose{\vec{R}}_i
\),
where \(\transpose{\vec{R}}_i\) is the \(i\)-th row of \(\transpose{\vec{R}}\).\footnote{%
    One might also compute this as a vector of dot-products between
    \(\vec{s}\) and the columns of \(\transpose{\vec{R}}\);
    but when using bitsliced arithmetic,
    it seems harder to exploit the weight of \(\vec{s}\)
    to reduce the cost of this approach.
}
Since the \(s_i\) are in \(\FF_3\),
no multiplication is really required:
computing the sum
amounts to adding \(\transpose{\vec{R}}_i\) into an accumulator if \(s_i = 1\),
subtracting it if \(s_i = 2\),
and ignoring it completely if \(s_i = 0\).
Skipping rows is an important optimization,
and not just to minimize field operations:
loading rows from memory (or external flash) is expensive.

But the vector \(\vec{s}\) has unusually few \(0\)s.
Indeed,
\(|\vec{e}| = w\),
which is quite close to \(n\) in the \Wave context,
so we expect \(|\vec{s}|\) to be proportionally close to \(k\).
At first, this
may appear frustrating: there are not many rows to skip
when computing \(\vec{s}\transpose{\vec{R}}\).
However,
we can turn this imbalance around
to create a relatively sparse sum
using the following propositions.

First,
for each \(\vec{s}\) in \(\FF_3^k\),
we define a vector \(\vec{\hat{s}}\) in \(\FF_3^k\)
by
\begin{equation}
    \label{eq:s-hat-def}
    \begin{rcases}
        \hat{s}_{2i} := s_{2i} + s_{2i+1}
        \\
        \hat{s}_{2i+1} := s_{2i} - s_{2i+1}
    \end{rcases}
    \quad
    \text{for}
    \quad
    0 \le i < k/2
    \,,
    \quad
    \text{and}
    \quad
    \hat{s}_{k-1} := s_{k-1}
    \text{ if \(k\) is odd}
    \,.
\end{equation}

\begin{proposition}
    \label{prop:s-hat-weight}
    Let \(\vec{s}\) be a vector in \(\FF_3^k\),
    and let \(\vec{\hat{s}}\) be defined as
    in~\eqref{eq:s-hat-def}.
    Then
    \[
        |\vec{\hat{s}}| + |\vec{s}|
        =
        \frac{3}{2}(k+\epsilon) - 3\delta
    \]
    where
    \[
        \delta = \#\{0 \le i < k/2 \mid s_{2i} = s_{2i+1} = 0\}
        \quad
        \text{and}
        \quad
        \epsilon 
        = 
        \begin{cases}
            0 & \text{if \(k\) is even};
            \\
            1 & \text{if \(k\) is odd and } s_{k-1} \not= 0;
            \\
            -1 & \text{if \(k\) is odd and } s_{k-1} = 0\,.
        \end{cases}
    \]
\end{proposition}
\begin{proof}
    See Appendix~\ref{appendix:proofs}.
\end{proof}

Proposition~\ref{prop:s-hat-weight}
shows that as \(|\vec{s}|\) grows,
\(|\vec{\hat{s}}|\) shrinks.
Indeed, in the limiting case where \(|\vec{s}| = k\),
we find \(|\vec{\hat{s}}| = \lceil{k/2}\rceil\).
Now we just need to find a way to verify using a vector-matrix product
involving \(\vec{\hat{s}}\) instead of \(\vec{s}\).

\begin{proposition}
    \label{prop:M}
    Let \(\vec{R}\) be in \(\FF_3^{(n-k)\times k}\).
    For every \(\vec{s}\) in \(\FF_3^k\),
    if \(\vec{\hat{s}}\) is defined
    as in~\eqref{eq:s-hat-def},
    then
    \begin{equation}
        \label{eq:vM-sR}
        \vec{s}\transpose{\vec{R}}
        =
        -\vec{\hat{s}}\vec{M}
    \end{equation}
    where \(\vec{M}\) is the matrix in \(\FF_3^{k\times(n-k)}\)
    whose rows are
    \[
        \begin{rcases}
            \vec{M}_{2i} := \transpose{\vec{R}}_{2i} + \transpose{\vec{R}}_{2i+1}
            \\
            \vec{M}_{2i+1} := \transpose{\vec{R}}_{2i} - \transpose{\vec{R}}_{2i+1}
        \end{rcases}
        \quad
        \text{for}
        \quad
        0 \le i < k/2
        \,,
        \quad
        \text{and}
        \quad
        \vec{M}_{k-1} := -\transpose{\vec{R}}_{k-1}
        \text{ if \(k\) is odd}
        \,.
    \]
\end{proposition}
\begin{proof}
    See Appendix~\ref{appendix:proofs}.
\end{proof}

Proposition~\ref{prop:M}
shows that
if we replace the public key \(\vec{R}\)
with \(\vec{M}\),
then the validation equation~\eqref{eq:intermediate-verification} becomes
\begin{equation}
    \label{eq:Wavelet-verification}
    |\vec{u}| + |\vec{s}| = w
    \,,
    \quad
    \text{where}
    \quad
    \vec{u} = \vec{h} + \vec{\hat{s}}\vec{M}
    \quad
    \text{with}
    \quad
    \vec{h} = \Hash(r\parallel m)
    \,.
\end{equation}
Clearly \(\transpose{\vec{R}}\)
can be recovered from \(\vec{M}\),
and \(\vec{s}\) from \(\vec{\hat{s}}\),
so solving~\eqref{eq:Wavelet-verification}
is equivalent to solving~\eqref{eq:intermediate-verification}.
But \(\vec{\hat{s}}\),
which can be rapidly computed from \(\vec{s}\) on the fly
using~\eqref{eq:s-hat-def},
has much lower weight by Proposition~\ref{prop:s-hat-weight}:
roughly speaking, we expect \(|\vec{s}|\) to be close to \(k\),
and so \(|\vec{\hat{s}}|\) should be close to \(k/2\).
This lets us compute the product \(\vec{\hat{s}}\vec{M}\),
and hence verify~\eqref{eq:Wavelet-verification},
in much less time than~\eqref{eq:intermediate-verification}
(for \Supertubos, nearly half);
and less time again than that required for
the equivalent \Wave signatures using~\eqref{eq:Wave-verification}.

We therefore take \(\vec{M}\) to be the public key in \Wavelet;
we verify \Wavelet signatures \((r,\vec{s})\)
using~\eqref{eq:Wavelet-verification} and~\eqref{eq:s-hat-def}.
Note that we do \emph{not} replace \(\vec{s}\) with \(\vec{\hat{s}}\) in the signature:
the weight condition \(w = |\vec{u}| + |\vec{s}|\)
is simpler than the corresponding condition involving \(|\vec{\hat{s}}|\),
and in any case the weight \(|\vec{s}|\) can be easily computed at the same time as
\(\vec{\hat{s}}\).

\begin{remark}
    The mapping \(\vec{s} \mapsto \vec{\hat{s}}\) can be recognised as
    the first step in a fast Hadamard transform.
    We might naturally ask if pushing further with the Hadamard
    transform on \(\vec{s}\) would produce even lower weights;
    it does not.  Instead, the weight of the image vector creeps back up
    again.  With the \Supertubos parameters,
    the expected \(|\vec{s}|/k\) is \(\approx 0.94\).
    One step of the fast Hadamard transform
    gives an expected \(|\vec{\hat{s}}|/k\) of \(\approx 0.55\);
    but the next step
    brings this expected ratio up to \(\approx 0.65\),
    and the following step to \(\approx 0.67\).
\end{remark}

\subsection{Signature compression}
\label{sec:high-level-compression}

A \Wavelet signature is a pair \((r,\vec{s})\),
with \(r\) in \(\{0,1\}^{2\lambda}\)
and \(\vec{s}\) in \(\FF_{3}^k\).
We must choose a binary encoding of \(\vec{s}\),
and at first glance it would appear
that this requires (at least) \(\log_2(3)k\) bits.
However, this ignores the fact that \(\vec{s}\) has very high weight,
and thus an entropy significantly less than \(\log_2(3) \approx 1.585\).

Concretely, if we use the \Supertubos parameters
where \(n = 8492\) and \(w = 7980\),
then the entropy is \(\approx 1.268\).
In theory, then, we can hope to encode \(\vec{s}\)
in as few as \(1.268k \approx 7109\) bits.
This is better than \(1.585k \approx 8884\) bits,
and much better than the \(1.585n \approx 13460\)
bits for \Wave signatures suggested
in~\cite{2019/Debris-Alazard--Sendrier--Tillich}
(``the order of 13 thousand bits'').

We show in~\S\ref{sec:MR-compression}
that we can approach this level of efficiency
using a simple Minimal Redundancy (MR) coding.
This variable-length coding yields variable-length signatures,
but on average we can compress \(\vec{s}\) to around \(7280\) bits.

\section{
    \Wavelet algorithms
}
\label{sec:scheme}
In this section, we present the details of our hashing,
Key Generation, Signing, and Verification algorithms.
The key generation and signing algorithms are \emph{not} constant-time,
and should only be run in trusted environments.
We leave the development of countermeasures against
side-channel attacks for \Wave and \Wavelet
for future work.

\subsection{Hashing to ternary vectors}
\label{sec:hash}

In order to sign any messages,
we need to define a hash function into \(\FF_3^{n-k}\).
The simplest option would be to hash into
\(\{0,1\}^{\lceil{\log_2(3)(n-k)}\rceil}\) using an XOF,
view the output as an integer, and ``ternarize'' the result
(i.e., read off its 3-adic coefficients)
to get the entries of the output in \(\FF_3^{n-k}\).
For realistic parameter sizes, though,
this is too slow and intensive:
for \Supertubos, for example,
this would mean repeated Euclidean division by 3
on a 4576-bit integer.
This is inconvenient on a PC,
and unrealistic on embedded platforms.

Our hash function, called \Hash (Algorithm~\ref{algo:hash})
combines the following three functions
to define a cryptographic hash function \(\{0,1\}^* \to \FF_3^{n-k}\):
\begin{itemize}
    \item
        \(\BinaryHash: \{0,1\}^* \to \{0,1\}^{2\lambda}\)
        is a cryptographic hash function.
    \item
        \(\Ternarize: \{0,1\}^*\times\ZZ_{\ge 0} \to \FF_3^*\)
        (Algorithm~\ref{algo:ternarize})
        views its input \(((x_0,x_1,\ldots),\tau)\) 
        as the vector of coefficients in the little-endian binary expansion of an
        integer \(x\), together with a length \(\tau\),
        and returns the ternary vector of length $\tau$
        representing the little-endian ternary expansion of \(x \bmod{3^\tau}\).
    \item
        \Expand: \(\{0,1\}^{2\lambda} \to \FF_{3}^{\tau}\)
        (Algorithm~\ref{algo:expand})
        is a pseudorandom function.
        \Expand applies an XOF (or a stream cipher) to its input to produce a long
        stream of pseudorandom bytes, which we view as integers in
        \([0,255]\).
        The non-negative integers less than \(3^5 = 243\)
        are in bijection with \(\FF_{3}^5\),
        so if a byte is less than 243 we convert it to an element of
        \(\FF_{3}^5\) with \Ternarize and concatenate it to the output;
        otherwise we skip the byte.
        We continue processing bytes until we have produced \(\tau\)
        elements of \(\FF_3\) (discarding the last few if \(\tau\) is not a multiple of~5).
        This process generates a distribution of vectors in \(\FF_{3}^{\tau}\)
        that is computationally indistinguishable from the uniform distribution
        (with respect to the security parameter).\footnote{%
            A true uniform distribution on \(\FF_{3}^{\tau}\)
            cannot be produced in this way when \(\tau\) is large, because there are only
            \(2^{2\lambda}\) possible stream outputs---and even for large \(\lambda\),
            the number of possible output vectors
            is limited by the size of the internal state of the XOF or
            stream cipher used to generate the stream.
            On the other hand,
            this process has the advantage of being simple and fast.
        }
\end{itemize}

\begin{algorithm}[ht]
    \caption{Hashing from binary data to ternary vectors for \Wavelet.
        We write \(\mu\) for \(\lfloor{2\lambda/\log_2(3)}\rfloor\).
    }
    \label{algo:hash}
    \Function{\Hash{$x$}}{
        \KwIn{%
            $x\in\{0,1\}^*$
        }
        \KwOut{%
            $\vec{s}\in\FF_{3}^{n-k}$
        }
        $h \gets$ \BinaryHash{$x$}
        \tcp*{Hash function $\{0,1\}^* \to \{0,1\}^{2\lambda}$}
        $\vec{t} \gets$ \Ternarize{$h$, $\mu$}
        \tcp*{\Ternarize: $\{0,1\}^{2\lambda} \to \FF_{3}^{\mu}$}
        $\vec{p} \gets$ \Expand{$h$, $n-k-\mu$}
        \tcp*{\Expand: $\{0,1\}^{2\lambda} \to \FF_{3}^{n-k-\mu}$}
        \Return{$\vec{t} \parallel \vec{p}$}
    }
\end{algorithm}

The collision- and preimage-resistance of \Hash
are derived from the properties of \(\BinaryHash\).
\Ternarize transcodes the output of \BinaryHash;
this composition maintains the security of \BinaryHash,
but is relatively slow to compute.
The composition of \Expand and \BinaryHash 
has weaker preimage and collision resistance,
but good pseudorandomness properties,
and is relatively fast to compute.
The concatenation of the two has the security of the strong hash,
and the good pseudorandomness of both.

\begin{algorithm}[ht]
    \caption{Converting integer values to ternary vectors of a specified length,
    corresponding to the little-endian ternary expansion of the input}
    \label{algo:ternarize}
    \Function{\Ternarize{$x$, $\tau$}}{
        \KwIn{$x \ge 0$ and $\tau > 0$}
        \KwOut{%
            $\vec{v} \in \FF_{3}^\tau$
            such that $x = \sum_{i=1}^{\tau}v_i3^{i-1}$,
            where the \(v_i\) are lifted to \{0,1,2\}
        }
        $(\vec{v},t) \gets ((),x)$
        \tcp*{\(\vec{v}\): empty vector over \(\FF_3\)}
        \For{$1 \le i \le \tau$}{
            $(\vec{v},t) \gets (\vec{v} \parallel (r), q)$
            \Where
            $(q,r) = (\lfloor{t/3}\rfloor, t \bmod{3})$
            \;
        }
        \Return{$\vec{v}$}
    }
\end{algorithm}

\begin{algorithm}[ht]
    \caption{%
        Expand a binary seed to a pseudo-random stream of
        ternary values.
        The random bytestream may be instantiated with an XOF or a
        stream cipher.
        The expected number of bytes drawn from the stream
        is \((256\tau)/(243\times5) \approx 0.21\tau\).
    }
    \label{algo:expand}
    \Function{\Expand{$h$, $\tau$}}{
        \KwIn{$h\in\{0,1\}^{2\lambda}$ and $\tau > 0$}
        \KwOut{$\vec{p}\in\FF_{3}^{\tau}$}
        \SetKwData{Stream}{stream}
        \Stream = random bytestream seeded with \(h\)
        \tcp*{E.g. secure XOF(\(h\))}
        $(\vec{p},r) \gets ((),\tau)$
        \tcp*{\(\vec{p}\): empty vector over $\FF_{3}$}
        \While{$r > 0$}{
            \(b \gets\) next byte from \Stream,
            viewed as an integer in \([0,255]\)
            \;
            \If{$b < 243$}{
                $(\vec{p},r) \gets (\vec{p}\parallel \Ternarize(b, \min(5,r)), r)$
                \;
            }
        }
        \Return{$\vec{p}$}
    }
\end{algorithm}

In practice, \(\BinaryHash\) could be SHA3-512,
while the bytestream in \Expand could be generated with an XOF such as SHAKE-256.
This minimises the code size, since both are built on the same Keccak permutation,
while maintaining appropriate domain separation.

\subsection{Key generation}
\label{sec:keygen}

Algorithm~\ref{algo:keygen} 
generates a \Wavelet keypair.
It is essentially the \Wave key-generation procedure
from~\cite{2019/Debris-Alazard--Sendrier--Tillich},
with the public key transformation at the end.
We have made no attempt to optimize this procedure
beyond the use of bitsliced and vectorized \(\FF_3\)-arithmetic.

Given 
an \((n/2-k_U)\times k_U\) matrix \(\vec{R}_U\),
an \((n/2-k_V)\times k_V\) matrix \(\vec{R}_V\),
and vectors \(\vec{a}\), \(\vec{b}\), \(\vec{c}\), and \(\vec{d}\) in
\(\FF_3^{n/2}\),
the subroutine \ParityCheckUV{$\vec{R}_U$, $\vec{R}_V$, $\vec{a}$, $\vec{b}$, $\vec{c}$, $\vec{d}$}
returns the $(n-k)\times n$ matrix
\[
     \Hsec
     =
     \begin{pmatrix}
         \vec{H}_{U}\vec{D} & - \vec{H}_U \vec{B}
         \\
         - \vec{H}_V \vec{C} &  \vec{H}_V \vec{A}
     \end{pmatrix}
     \quad
     \text{where}
     \quad
     \begin{cases}
         \vec{H}_U := (\vec{I}_{n/2 - k_{U}}\mid \vec{R}_U)
         \,,
         \\
         \vec{H}_V := (\vec{I}_{n/2 - k_{V}}\mid \vec{R}_V)
         \,,
     \end{cases}
\]
and
\[
    \vec{A} := \Diag(\vec{a})\,,
    \qquad
    \vec{B} := \Diag(\vec{b})\,,
    \qquad
    \vec{C} := \Diag(\vec{c})\,,
    \qquad
    \vec{D} := \Diag(\vec{d})\,;
\]
here, $\Diag(\vec{x})$ denotes the $n/2 \times n/2$ diagonal matrix with
diagonal entries given by $\vec{x}$.

\begin{algorithm}[ht]
    \caption{\Wavelet Key Generation.
        This algorithm also serves for classic \Wave key generation
        if we stop after Line~\ref{algo:keygen:Wave-stop-line}
        and return \(\sk\) and \(\pk =(\vec{I}_{n-k}|\vec{R})\).
    }
    \label{algo:keygen}
    \Function{\KeyGen{}}{
        \KwParameters{$n,w,k_{U},k_{V},k$}
        \KwOut{$(\sk, \pk)$}
        $\vec{R}_{U} \Unif\F_3^{(n/2-k_U)\times k_U}$
        \;
        $\vec{R}_{V}\Unif\F_{3}^{(n/2-k_V)\times k_V}$
        \;
        $\vec{a} \Unif \left(\F_{3}\backslash\{0\}\right)^{n/2}$
        \;
        $\vec{b} \Unif \F_{3}^{n/2}$
        \;
        $\vec{c} \Unif \left(\F_{3}\backslash\{0\}\right)^{n/2}$
        \;
        $\vec{d} \leftarrow (d_{i})_{0 \leq i < n/2}$ \Where $d_{i} \Unif \F_{3} \backslash \{ b_{i}c_{i}a_{i}^{-1} \}$
        for \(0 \le i < n/2\)
        \;
        $\Hsec \gets$ \ParityCheckUV{%
            $\vec{R}_{U}$,
            $\vec{R}_{V}$,
            $\vec{a}$,
            $\vec{b}$,
            $\vec{c}$,
            $\vec{d}$
        }
        \;
        %
            $\pisec\sample\Sym{n}$
            \;
            $((\mathbf{I}_{n-k}\mid \vec{R}),\pisec) \gets$ \GaussElim{$\Hsec$, $\pisec$}
            \;
        %
        $\sk\gets (\vec{R}_{U},\vec{R}_{V},\vec{a},\vec{b},\vec{c},\vec{d},\pisec)$
        \label{algo:keygen:Wave-stop-line}
        \;
        \(\vec{M} = \vec{0}^{k\times(n-k)}\) 
        \;
        \For(\tcp*[f]{%
            Build first \(2\lfloor{k/2}\rfloor\) rows 
            of \(\vec{M} = \vec{T}\transpose{\vec{R}}\)}){\(0 \le i < (k-1)/2\)}{
            \(\vec{M}_{2i} \gets (\transpose{\vec{R}})_{2i} + (\transpose{\vec{R}})_{2i+1}\)
            \tcp*{Sum of rows}
            \(\vec{M}_{2i+1} \gets (\transpose{\vec{R}})_{2i} - (\transpose{\vec{R}})_{2i+1}\)
            \tcp*{Difference of rows}
        }
        \If(\tcp*[f]{Fill in last row if necessary}){\(k\) is odd}{
            \(\vec{M}_{k-1} \gets -(\transpose{\vec{R}})_{k-1}\)
            \tcp*{The negative is important}
        }
        \(\pk \gets \vec{M}\)
        \tcp*{\(\vec{M} \in \FF_3^{k\times(n-k)}\)}
        \Return{$(\sk,\pk)$}
    }
\end{algorithm}

The auxiliary function $\GaussElim$
is given in Algorithm~\ref{algo:GE}. For this function,
we need to ``split'' into two functions, that is,
in our scheme we use it to generate a submatrix
with certain rank and for this, we need to have
the support as part of our secret key thus
we changed a little to know which rows are
the pivots.

\begin{algorithm}[ht]
    \caption{GaussElim}
    \label{algo:GE}
    \KwIn{%
        $\vec{H}\in \F_3^{R\times C}$
        and
        $\pi \in \Sym{D}$ with \(D \le C\)
    }
    \KwOut{%
        $\vec{H} =  (\vec{I}_R | \vec{M}) \in \F_3^{R\times C}$
        where $\vec{M} \in \F_3^{R\times C-R}$,
        and $\pi' \in\Sym{D}$
        where the first \(R\) positions are the pivots.
    }
    \SetKwData{Pivot}{pivot}
    \SetKwData{Nonpivot}{nonpivot}
    \Function{\GaussElim{$\vec{H}$, $\pi$}}{
        $\Pivot \gets ()$
        \;
        $\Nonpivot \gets ()$
        \;
        $(r, c) \gets (0, 0)$
        \;
        \While{$r < R$ \And $c < D$}{
            \((\vec{H},P) \gets\) \ElimSingle{$\vec{H}$, $r$, $\pi(c)$}
            \tcp*{Algorithm~\ref{algo:ES}}
            \uIf{$P$}{
                $(\Pivot,r) \gets (\Pivot\parallel(\pi(c)), r+1)$
            }
            \Else{
                $\Nonpivot \gets \Nonpivot\parallel(\pi(c))$
            }
            $c \gets c + 1$
            \;
        }
        $\pi' \gets \Pivot \parallel \Nonpivot \parallel (\pi(c),\dots,\pi(D-1))$
        \;
        \Return{$(\pi'(\vec{H}),\pi')$}
    }
\end{algorithm}

Algorithm~\ref{algo:ES} gives the elimination of a single row,
for this we check if the element is not $0$ then we perform the
operations.
Later,
we swap rows to put in the correct ``rank'', and
do the sum and subtraction necessary to just let
the $(r,j)$-th element as not zero.

\begin{algorithm}[ht]
    \caption{Attempted Gaussian elimination on a single column of a
        matrix.
    }
    \label{algo:ES}
    \KwIn{%
        $\vec{H}\in \F_3^{R\times C}$,
        row index \(0 \le r < R\),
        column index \(0 \le c < C\)
    }
    \KwOut{%
        \((\vec{H},P)\)
        where \(P\) is \True/\False if a pivot was/was not found in
        column $c$.
    }
    \Function{\ElimSingle{$\vec{H}$, $r$, $c$}}{
        \(p \gets r\)
        \;
        \While(\tcp*[f]{Search for pivot position}){\(p < R\) \And \(\vec{H}_{p,c} = 0\)}{
            \(p \gets p + 1\)
        }
        \If(\tcp*[f]{Pivot not found}){\(p = R\)}{
            \Return{\((\vec{H}, \text{\False})\)}
        }
        \If{$ \vec{H}_{p,c} = 2$}{
            \(\vec{H} \gets\) \Normalize{$\vec{H}$, $p$}
            \tcp*{Negate \(p\)-th row}
        }
        \(\vec{H} \gets\) \SwapRows{$\vec{H}$, $p$, $r$}
        \tcp*{Move pivot to \(r\)-th row}
        \For{$i$ \In $(0,\ldots,r-1,r+1,\ldots,R-1)$}{
            \uIf{$\vec{H}_{i,c} = 1$}{
                \(\vec{H} \gets\) \SubtractRow{$\vec{H}$, $i$, $r$}
                \tcp*{Subtract \(r\)-th row from \(i\)-th row}
            }
            \ElseIf{$\vec{H}_{i,c} = 2$}{
                \(\vec{H} \gets\) \SumRow{$\vec{H}$, $i$, $r$}
                \tcp*{Add \(r\)-th row to \(i\)-th row}
            }
        }
        \Return{\((\vec{H}, \text{\True})\)}
    }
\end{algorithm}

\subsection{Signing}
\label{sec:sign}

Algorithm~\ref{algo:sign} defines the \Wavelet signing process,
which is essentially the same as \Wave signing.
We have made no attempt to optimize this procedure
beyond the use of bitsliced and vectorized \(\FF_3\)-arithmetic.

The main objective of the \Wave signing procedure is
to find a vector $\vec{e}\in\F_{3}^{n}$ such that
\begin{equation}\label{eq:e}
    |\vec{e}| = w\,,
    \quad \text{and} \quad
    \vec{e}\transpose{(\mathbf{I}_{n-k}\mid \vec{R})} = \vec{h}
\end{equation}
where $\vec{h}$ is the salted hash of the message to be signed,
and $\vec{R}\in\F_{3}^{(n-k)\times k}$ is the public key.
The public key is generated in such a way that
$$
    (\mathbf{I}_{n-k}\mid \vec{R})
    =
    \vec{S}\pisec\left(\Hsec \right)
    \quad \text{with}\quad
    \Hsec
    =
    \begin{pmatrix}
	    \vec{H}_{U} \vec{D} & - \vec{H}_U \vec{B}
        \\ \hline
	    - \vec{H}_V \vec{C} &  \vec{H}_V \vec{A}
    \end{pmatrix}
$$
where $\vec{S}$ is a nonsingular matrix
corresponding to the Gaussian elimination putting $\pisec(\Hsec)$ in row-reduced form.
Therefore,
the solution $\vec{e}$ in \eqref{eq:e}
is also a solution of the system
$$
    \pisec\left(\Hsec\right)\transpose{\vec{e}}
    =
    \vec{S}^{-1}\transpose{\vec{h}}
    =
    \pisec(\Hsec)\transpose{(\vec{h},\mathbf{0}_{k})}
    \,.
$$

We therefore aim to find some vector $\vec{e}'$ such that
\begin{equation}
    \label{eq:ep}
    |\vec{e}'| = w
    \quad \text{and} \quad
    \vec{e}'\transpose{\Hsec} = \vec{h}'
    \quad \text{where}\quad
    \vec{h}' \eqdef \pisec(\Hsec)\transpose{(\vec{h},\mathbf{0}_{k})}
\end{equation}
and output $\vec{e} = \pisec(\vec{e}')$.
Using the special structure of $\Hsec$,
we see that
\begin{equation}\label{eq:systToSolve}
    \Hsec\transpose{\vec{e}'}
    =
    \vec{h}'
    \iff
    \begin{cases}
        \transpose{\vec{e}}_{U}\vec{H}_{U} = \vec{h}^{U} \,,
        \\
        \transpose{\vec{e}}_{V}\vec{H}_{V} = \vec{h}^{V} \,.
    \end{cases}
\end{equation}
where $\vec{h}' = \vec{h}^{U}\parallel\vec{h}^{V}$
and
$
    \vec{e}'
    =
    \big(
        \vec{e}_{U}\vec{A} + \vec{e}_{V}\vec{B}
    \big)
    \parallel
    \big(
        \vec{e}_{U}\vec{C} + \vec{e}_{V}\vec{D}
    \big)
$
(note that $\vec{a}$, $\vec{b}$, $\vec{c}$, and $\vec{d}$ are generated
in such a way that
$
    (\vec{x},\vec{y})
    \mapsto
    \big(
        \vec{x}\vec{A} + \vec{y}\vec{B},
        \vec{x}\vec{C} + \vec{y}\vec{D}
    \big)
$
is a bijection).

To find $\vec{e}$,
we solve both linear systems in~\eqref{eq:systToSolve}
using the \Wave decoder originally described
in~\Wave~\cite{2019/Debris-Alazard--Sendrier--Tillich}.
Pseudocode for the decoder is given in Algorithm~\ref{algo:dec}
in Appendix~\ref{ap:decoder}.
There are two main steps:
\begin{enumerate}
	\item
        First,
        compute any solution $\vec{e}_{V}$ of the undetermined linear system $\vec{e}_{V}\transpose{\vec{H}}_{V}=\vec{h}^{V}$
	\item
        Then,
        compute a \emph{particular} solution $\vec{e}_{U}$ of the second undetermined linear system.
\end{enumerate}

For the second step,
the \Wave decoder uses a generalized version of Prange's decoder~\cite{prange62}.
Roughly speaking,
given $\vec{H}_{U}$ and $\vec{h}^{U}$,
when solving the linear system
$\vec{e}_{U}\transpose{\vec{H}_{U}} = \vec{h}^{U}$ (with $n/2$ unknowns and $n/2-k_U$ equations),
we are free to select values of $\vec{e}_{U}$ on $k_{U}$ coordinates:
let's say positions $(0,\dots,k_{U}-1)$
(though these $k_U$ coordinates can be chosen almost anywhere).
Since we are looking for a solution
$
    \vec{e}'
    =
    \left(\vec{e}_{U}\vec{A} + \vec{e}_{V}\vec{B},\vec{e}_{U}\vec{C} + \vec{e}_{V}\vec{D}\right)
$
of large weight, the best strategy is to choose
$(e_{U})_0,\dots,(e_{U})_{k_{U}-1}$
such that
\[
	(\vec{e}_{U}\vec{A} + \vec{e}_{V}\vec{B})_j \neq 0
    \quad
    \text{and}
    \quad
	(\vec{e}_{U}\vec{C} + \vec{e}_{V}\vec{D})_j \neq 0
    \quad
    \text{for all}
    \quad
    0 \le j < k_{U}
    \,.
\]
This is always possible over \(\FF_3\)
given our choice of $\vec{A}$, $\vec{B}$, $\vec{C}$, and $\vec{D}$.
The other $n - 2k_{U}$ coordinates of $\vec{e}'$ will be uniformly
distributed over $\F_{3}$,
because they are obtained from coordinates of $\vec{e}_{U}$ that we have
no control over when solving the random square linear system.
Therefore, we typically expect $|\vec{e}'| = 2k_{U} + 2/3(n-2k_{U})$.
Choosing $w =  2k_{U} + 2/3(n-2k_{U})$
ensures that our decoder will succeed after a polynomial number of steps.
Now, by carefully choosing $k_{U}$,
our decoder solves a trapdoor problem:
finding $\vec{e}$ of weight $w$ such that $\vec{e}\transpose{\Hpub} = \vec{h}$
is hard unless we know the hidden structure (via $\pisec$) of $\Hpub$.

\paragraph{Rejection Sampling.}
The procedure described above has a serious drawback:
signatures $\vec{e}$ may leak information about $\pisec$ (see~\cite[5.1]{DebrisST17b}).
To avoid this, the outputs of the signing algorithm must to be very close to uniformly distributed over words of Hamming weight $w$.
By carefully adjusting with some rejection sampling the Hamming weights of $\vec{e}_{V}$ and $\vec{e}_{U}$,
we can meet this property by still producing solutions of weight $w$ for which it is hard to solve the decoding problem for this weight.
See~\cite[5.1]{2019/Debris-Alazard--Sendrier--Tillich} for more details.

\begin{algorithm}[ht]
    \caption{The \Wavelet signing algorithm.
        This algorithm also produces classic \Wave signatures
        if we stop at Line~\ref{alg:sign:Wave-return}
        and return \((\vec{e},r)\).
    }
    \label{algo:sign}
    \Function{\Sign{$m$, $\sk$}}{
        \KwIn{%
            $m \in\{0,1\}^*$,
            \(\sk= (\vec{R}_{U},\vec{R}_{V},\vec{a},\vec{b},\vec{c},\vec{d},\pisec)\)
        }
        \KwOut{$\vec{e}\in\F_{3}^n$, $|\vec{e}| = w$}
        \(\Hsec \gets\)
        \ParityCheckUV{%
            $\vec{R}_{U}$,
            $\vec{R}_{V}$,
            $\vec{a}$,
            $\vec{b}$,
            $\vec{c}$,
            $\vec{d}$
        }
        \;
        $r \Unif\{0,1\}^{2\lambda}$
        \tcp*{salt}
        $\vec{h}\gets \Hash(r\parallel m)$
        \tcp*{syndrome in \(\F_3^{n-k}\)}
        \(
            \vec{x}
            \gets
            (\vec{h},\mathbf{0}_{k})
            \transpose{\pisec(\Hsec)}
        \)
        \tcp*{syndrome adjusted for decoding}
        $\vec{y}\gets$
        \Decode{%
            $\vec{x}$,
            $\vec{R}_U$,
            $\vec{R}_V$,
            $\vec{a}$,
            $\vec{b}$,
            $\vec{c}$,
            $\vec{d}$
        }
        \;
        \(\vec{e} \gets \pisec(\vec{y})\)
        \tcp*{\((\vec{e}, r) =\) \Wave signature}
        \label{alg:sign:Wave-return}
        \(\vec{s} \gets (e_{n-k},\ldots,e_{n-1})\)
        \tcp*{Truncate to last \(k\) entries}
        \Return{$(r,\vec{s})$}
    }
\end{algorithm}

\subsection{Verification}
\label{sec:verify}

Algorithm~\ref{algo:stverif-M}
verifies \Wavelet signatures \((r,\vec{s})\)
under public keys \(\vec{M}\)
following the derivation of~\S\ref{sec:Wavelet-signatures},
and using the verification equation~\eqref{eq:Wavelet-verification}:
\[
    |\vec{u}| + |\vec{s}| = w
    \,,
    \quad
    \text{where}
    \quad
    \vec{u} = \vec{h} + \vec{\hat{s}}\vec{M}
    \quad
    \text{with}
    \quad
    \vec{h} = \Hash(r\parallel m)
    \,.
\]
We iterate over the signature vector \(\vec{s}\),
computing the entries of the sparser vector \(\vec{\hat{s}}\)
on the fly using~\eqref{eq:s-hat-def},
and accumulating the corresponding rows of the public key~\(\vec{M}\).
Concretely, if \(w/n\) is around \(0.94\)
(as it is in \Supertubos),
then we expect Algorithm~\ref{algo:stverif-M}
to use around \(0.56k\) row vector operations.
Notice that there is only a single pass over the rows of \(\vec{M}\),
so Algorithm~\ref{algo:stverif-M}
is compatible with applications where the public key is streamed.

\begin{algorithm}[ht]
    \caption{Verification for \Wavelet signatures.
        The entries of \(\vec{\hat{s}}\) are constructed from \(\vec{s}\) on
        the fly using~\eqref{eq:s-hat-def},
        and zero entries are skipped.
    }
    \label{algo:stverif-M}
    \Function{\STVerify{$m$, $\sigma$, \(\vec{M}\)}}{
        \KwIn{%
            $m\in\{0,1\}^*$,
            $\sigma = (r,\vec{s})\in\{0,1\}^{2\lambda}\times\F_3^{k}$,
            $\vec{M}\in\F_3^{k\times (n-k)}$
        }
        \KwOut{True or False}
        $\vec{x} \gets \Hash(r\parallel m)$
        \;
        \For(\tcp*[f]{Handle first \(2\lfloor{k/2}\rfloor\) entries of \(\vec{s}\) in pairs}){\(0 \le i < (k-1)/2\)}{
            \((\hat{s}_{2i},\hat{s}_{2i+1}) \gets (s_{2i}+s_{2i+1},s_{2i}-s_{2i+1})\)
            \tcp*{\((2i,2i+1)\)-th entries of \(\vec{s}\vec{T}\)}
            \uIf{\(\hat{s}_{2i} = 1\)}{
                \(\vec{x} \gets \vec{x} + \vec{M}_{2i}\)
                \tcp*{Add \(2i\)-th row of \(\vec{M}\)}
            }
            \ElseIf{\(\hat{s}_{2i} = 2\)}{
                \(\vec{x} \gets \vec{x} - \vec{M}_{2i}\)
                \tcp*{Subtract \(2i\)-th row of \(\vec{M}\)}
            }
            \uIf{\(\hat{s}_{2i+1} = 1\)}{
                \(\vec{x} \gets \vec{x} + \vec{M}_{2i+1}\)
                \tcp*{Add \((2i+1)\)-th row of \(\vec{M}\)}
            }
            \ElseIf{\(\hat{s}_{2i+1} = 2\)}{
                \(\vec{x} \gets \vec{x} - \vec{M}_{2i+1}\)
                \tcp*{Subtract \((2i+1)\)-th row of \(\vec{M}\)}
            }
        }
        \If(\tcp*[f]{Handle last entry if necessary}){\(k\) is odd}{
            \uIf{\(s_{k-1} = 1\)}{
                \(\vec{x} \gets \vec{x} + \vec{M}_{k-1}\)
            }
            \ElseIf{\(s_{k-1} = 2\)}{
                \(\vec{x} \gets \vec{x} - \vec{M}_{k-1}\)
            }
        }
        \Return{\(|\vec{s}| + |\vec{x}| = w\)}
        \;
    }
\end{algorithm}

\section{
    Encodings for keys and signatures
}
\label{sec:encodings}

\subsection{Basic vector encodings and key sizes}

Mathematically, \Wavelet works over \(\FF_3\);
but in reality, we must operate on binary values.
In practice, we will work with two binary encodings
for vectors over \(\FF_3\).

The \emph{compact} representation
views a group of \(s\) trits as the coefficients
of the ternary expansion of an integer in \([0,3^s)\),
and outputs the binary coefficients of the same integer.
This representation achieves 1.6 bits per trit,
or 5 trits per byte\footnote{%
    The compact representation packs 5 trits into a byte,
    or 20 into a 32-bit word,
    or 40 into a 64-bit word;
    so packing into words
    yields no improvement in space efficiency.
    Encoding to and decoding from words requires more operations on larger integers,
    too---so we prefer encoding to simple bytes.
}---which is very close to the optimum
of \(\log_2(3) \approx 1.585\) bits per trit.

For efficient \(\FF_3\)-vector arithmetic
(see~\S\ref{sec:implementation}),
we use a \emph{bitsliced} representation.
On a \(\beta\)-bit architecture,
we store a group of \(\beta\) trits in a pair of machine words:
the \(i\)-th trit is encoded using the \(i\)-th bits of both words.
This achieves 2 bits per trit, or 4 trits per byte.

The public key is essentially
a \(k\times(n-k)\) random matrix over \(\FF_3\).
In theory, it can be encoded in \(\log_2(3)k(n-k) \approx 1.585k(n-k)\) bits.
In practice, the compact representation requires \(1.6k(n-k)\) bits,
while the bitsliced representation requires \(\approx 2k(n-k)\) bits.

Technically, the private key can be stored in \(\lambda\) bits,
since it suffices to store the seed used to randomly generate
the matrices
\(\vec{R}_U\), \(\vec{R}_V\), \(\vec{a}\),
\(\vec{b}\), \(\vec{c}\), and \(\vec{d}\)
and the permutation \(\pi_{sk}\)
in Algorithm~\ref{algo:keygen}.
However, this means recomputing
them every time we sign a message.
In practice, then, we want to store the ``expanded'' private key
\((\vec{R}_U,\vec{R}_V,\vec{a},\vec{b},\vec{c},\vec{d},\pi_{sk})\).
The matrices $\vec{R}_{U}$ and $\vec{R}_{V}$ have sizes
\((\frac{n}{2} - k_U)\times\frac{n}{2}\),
and \((\frac{n}{2} - k_V)\times\frac{n}{2}\),
respectively,
while \(\vec{a}\), \(\vec{b}\), \(\vec{c}\), and \(\vec{d}\)
each have length \(\frac{n}{2}\).
We therefore need to store \((n - k + 4)(n/2)\) trits,
along with the permutation \(\pi_{sk}\),
which is a sequence of \(n\) integers of size \(\lceil{\log_2(n)}\rceil\).

\begin{table}
    \caption{Binary encodings of vectors over \(\FF_3\).}
    \label{tab:encodings}
    \centering
    \begin{tabular}{r|cc|rr}
        \toprule
        \multirow{2}{*}{Encoding} & Size & Rate
        & \multicolumn{2}{c}{\Supertubos key sizes} \\
        & (bits per trit) & (trits per byte) & Public & Expanded private \\
        \midrule
        \emph{(Optimal)} & \(\mathit{\log_2(3) \approx 1.585}\) & \(\mathit{\approx 5.047}\)
            & \(\approx \mathit{3131}~\si{\kilo\byte}\)
            & \(\approx \mathit{1350}~\si{\kilo\byte}\)
        \\
        \midrule
        Compact & \(1.6\) & 5 & \(\approx 3161~\si{\kilo\byte}\) & $\approx 1363~\si{\kilo\byte}$ \\
        \midrule
        Bitsliced & \(2\) & 4 & \(\approx 3951~\si{\kilo\byte}\) & $\approx 1703~\si{\kilo\byte}$ \\
        \bottomrule
    \end{tabular}
\end{table}

Table~\ref{tab:encodings} summarizes the performance of each encoding,
and gives concrete key sizes using the \Supertubos parameters.
We prefer the compact representation for public key transmission
(to minimize bandwidth),
but for local storage we prefer the bitsliced representation:
the 20\% space overhead
is a reasonable price to pay
to avoid the cost of transcoding from the compact to the bitsliced representation
before each vector operation.

\subsection{Compressed signature encodings}
\label{sec:MR-compression}

A \Wavelet signature is a pair \((r,\vec{s})\),
where \(r\) is a \(2\lambda\)-bit salt
and \(\vec{s}\) is a vector in \(\FF_3^k\).
As we mentioned in~\S\ref{sec:high-level-compression},
there are very few \(0\)s in \(\vec{s}\),
so the entropy is relatively low,
and we can hope to exploit this to compress signatures
to substantially less than \(\log_2(3)k\) bits.

For \Supertubos,
the probability that a given entry is \(1\) is $0.47$;
the probability of \(2\) is also \(0.47\);
and the probability of \(0\) is only \(0.06\).
This means \(1.267\) bits of entropy,
which is about 20\% smaller than \(\log_2(3) \approx 1.585\);
theoretically, we can compress an average \Supertubos \Wavelet signature
down to 7102 bits (or 888 bytes).

Simple Huffman encoding~\cite{huff_encoding}
brings us close to the theoretical limit,
as it takes advantage of these frequencies.
Huffman coding is a
compression technique that uses the frequencies of the
characters in a source alphabet to generate a binary
tree whose leaves are the characters. The code for
each character is defined by the path from the root to
the character's leaf: when take the left edge, we add a
$0$ to the code, and when we take the right edge,
we add a $1$, or vice-versa. The idea is to build
the tree bottom-up, putting the least frequent elements
first so that they stay on the lower levels, giving them
longer codes.

We will use the slightly more general framework of minimal redundancy
prefix codes (MR codes)~\cite{1997/Moffat--Turpin,2019/Moffat}.
We have two versions of the compression
algorithm: the first, \emph{\DuoHuff}, takes \(\FF_3^2\) as its source alphabet
(encoding two trits at a time);
the second, \emph{\TriHuff}, takes \(\FF_3^3\) as its source alphabet
(encoding three trits at a time).
Table~\ref{tab:duo-Huffman} shows
the binary outputs for these encodings,
derived using the usual MR code system
based on the average symbol frequencies for \Supertubos \Wavelet signatures.
Observe that \DuoHuff and \TriHuff are \emph{prefix-free}:
no output code is a prefix of any other.

To compress,
we break the signature vector \(\vec{s}\)
into a sequence of pairs (or triples),
and output the concatenation of the corresponding codes.
For decompression, we look at our input (a sequence of bytes)
as a stream of bits; the prefix-free property of the \DuoHuff and
\TriHuff codes
allows us to quickly and unambiguously recognise codewords,
and output the corresponding pairs (or triples) of trits,
using the algorithm of~\cite{1997/Moffat--Turpin}.
The decoding algorithm requires a few pre-computed lists
for fast lookup;
we give the relevant data, along with
details on the compression and decompression algorithms,
in Appendix~\ref{ap:comp}.

\begin{table}
    \caption{\DuoHuff and \TriHuff encodings for \Supertubos \Wavelet signatures}
    \label{tab:duo-Huffman}
    \label{tab:trio-Huffman}
    \centering
    \begin{tabular}{r|l||r|l|r|l|r|l}
        \toprule
        \multicolumn{2}{c||}{\textbf{\DuoHuff}}
        &
        \multicolumn{6}{c}{\textbf{\TriHuff}}
        \\
        Pair & Code &
        Triple & Code &
        Triple & Code &
        Triple & Code
        \\
        \midrule
        \((2,2)\) & \texttt{00}
        &
        \((2,2,2)\) & \texttt{000}
        &
        \((2,0,2)\) & \texttt{110101}
        &
        \((1,0,1)\) & \texttt{111110}
        \\
        \((2,1)\) & \texttt{01}
        &
        \((2,2,1)\) & \texttt{001}
        &
        \((0,2,2)\) & \texttt{110110}
        &
        \((0,1,1)\) & \texttt{1111110}
        \\
        \((1,2)\) & \texttt{100}
        &
        \((2,1,2)\) & \texttt{010}
        &
        \((2,1,0)\) & \texttt{110111}
        &
        \((2,0,0)\) & \texttt{111111100}
        \\
        \((1,1)\) & \texttt{101}
        &
        \((1,2,2)\) & \texttt{011}
        &
        \((2,0,1)\) & \texttt{111000}
        &
        \((0,2,0)\) & \texttt{1111111010}
        \\
        \((2,0)\) & \texttt{1100}
        &
        \((2,1,1)\) & \texttt{100}
        &
        \((0,2,1)\) & \texttt{111001}
        &
        \((0,0,2)\) & \texttt{1111111011}
        \\
        \((1,0)\) & \texttt{1101}
        &
        \((1,2,1)\) & \texttt{1010}
        &
        \((1,2,0)\) & \texttt{111010}
        &
        \((1,0,0)\) & \texttt{1111111100}
        \\
        \((0,2)\) & \texttt{1110}
        &
        \((1,1,2)\) & \texttt{1011}
        &
        \((1,0,2)\) & \texttt{111011}
        &
        \((0,1,0)\) & \texttt{1111111101}
        \\
        \((0,1)\) & \texttt{11110}
        &
        \((1,1,1)\) & \texttt{1100}
        &
        \((0,1,2)\) & \texttt{111100}
        &
        \((0,0,1)\) & \texttt{1111111110}
        \\
        \((0,0)\) & \texttt{11111}
        &
        \((2,2,0)\) & \texttt{110100}
        &
        \((1,1,0)\) & \texttt{111101}
        &
        \((0,0,0)\) & \texttt{1111111111}
        \\
        \bottomrule
    \end{tabular}
\end{table}

Since Huffman and MR coding is variable-length,
compressed \Wavelet signatures have variable length.
An extra level of variation is introduced by the fact
that the supposed symbol probabilities are not exact.
For example,
while we know that there are exactly \(n-w\) zeroes
in a full-length \Wave signature vector,
when we move to truncated \Wavelet signatures
the number of zeroes follows a binomial distribution
centred on \(k-wk/n\).

Figure~\ref{fig:compression}
shows the lengths of 22000 random \Supertubos \Wavelet signatures
compressed with the \DuoHuff and \TriHuff encodings.
We see that experimentally,
\DuoHuff-compressed \Supertubos \Wavelet signature vectors
fit in 950 bytes on the average (within 5\% of the theoretical optimum),
and almost always less than 970 bytes;
\TriHuff-compressed \Supertubos \Wavelet signature vectors
fit in 910 bytes on the average (within 2.5\% of the optimum),
and almost always less than 930 bytes.
Once we have included 32 bytes worth of salt,
\Supertubos \Wavelet signatures fit comfortably within one kilobyte.

\begin{figure}[h!]
    \caption{Frequencies of \DuoHuff and \TriHuff encoding lengths
    (in bytes) for $22\,000$ random \Supertubos \Wavelet signatures.}
    \label{fig:compression}
    \centering
    \includegraphics[width=\linewidth]{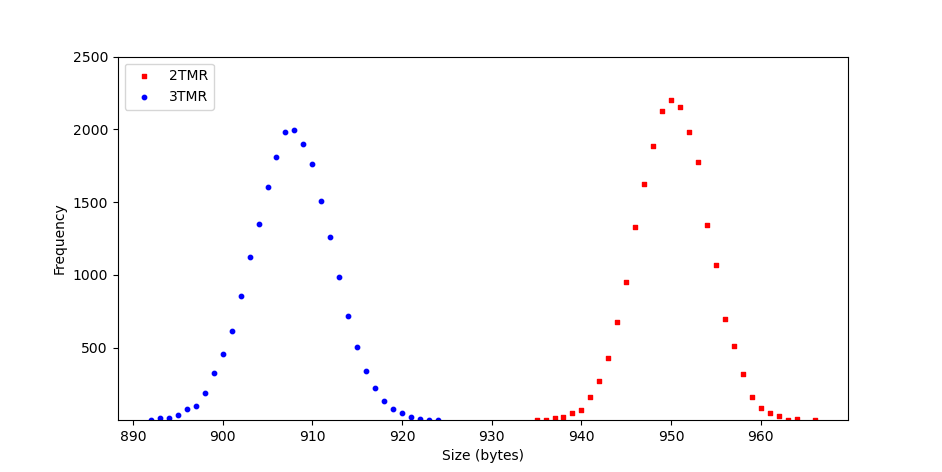}
\end{figure}

\section{
    Software implementations
}
\label{sec:implementation}
This section describes the techniques we used
to put the algorithms of~\S\ref{sec:scheme}
into practice on an
Intel\textsuperscript{\tiny\textregistered} Core\textsuperscript{TM}
platform with AVX,
and also on an ARM Cortex-M4 platform.

\subsection{Efficient ternary vector arithmetic}
\label{sec:ternary}

Our implementation
uses the bitsliced \(\FF_3\)-vector arithmetic
described in~\cite{Coolsaet};
the underlying
\(\FF_3\)-arithmetic
is the \emph{Type-2} representation
of~\cite[\S4.2]{2008/Kawahara--Aoki--Takagi}.
We use \(\oplus\), \(\&\), \(|\), and $\neg$
to denote logical XOR, AND, OR, and NOT, respectively.

Each element $a$ in $\FF_3$ is represented using
a pair of bits:
\[
    \FF_3 \ni a \longleftrightarrow (a_h, a_l) \in \{0,1\}^2
    \,;
\]
the encoding
of~\cite[\S4.2]{2008/Kawahara--Aoki--Takagi}
uses
\begin{align*}
    0 & \longleftrightarrow (0,0)
    \,,
    &
    1 & \longleftrightarrow (0,1)
    \,,
    &
    2 & \longleftrightarrow (1,1)
    \,.
\end{align*}
Additions and subtractions can be computed
at a cost of $7$ logical operations each
using the identities
\begin{align}
    \label{eq:addition_f_3}
    \big((a+b)_h,(a+b)_l\big)
    & =
    \big(
        (a_l \oplus b_h) \& (a_h \oplus b_l)
        \,,
        (a_l \oplus b_l) | ((a_h \oplus b_l) \oplus b_h)
    \big)
    \,,
    \\
    \label{eq:sub_type2_f_3}
    \big((a-b)_h,(a-b)_l\big)
    & =
    \big(
        (a_l \oplus (b_h \oplus b_l)) \& (a_h \oplus b_l)
        \,,
        (a_l \oplus b_l) | (a_h \oplus b_h)
    \big)
    \,;
\end{align}
pure negations can be computed using
\begin{equation}\label{eq:neg_f_3}
    \big((-a)_h, (-a)_l\big)
    =
    \big(a_h \oplus a_l, a_l\big)
    \,.
\end{equation}

This representation lends itself well to bitslicing~\cite{Biham97a}.
Working on a platform with \(\beta\)-bit machine words,
we break \(\FF_3\)-vectors down into a sequence of subvectors
\(\vec{a}\) in \(\FF_3^\beta\),
each of which is encoded as a pair of \(\beta\)-bit words:
\[
    \FF_3^\beta \ni \vec{a}
    \longleftrightarrow
    ( \vec{a}_h , \vec{a}_l )
    :=
    \big(
        \big((a_0)_h,\ldots,(a_{\beta-1})_h\big)
        ,
        \big((a_0)_l,\ldots,(a_{\beta-1})_l\big)
    \big)
    \,.
\]
The same sequences of logical operations
defining the addition, negation, and subtraction formul\ae{} above,
applied to the machine words \(\vec{a}_h\) and \(\vec{a}_l\),
thus compute addition, negation, and subtraction in parallel on \(\beta\) bits at a time.

\paragraph{AVX instructions.}
Going further in this direction, many computer architectures have special
instructions to vectorize logical operations across multiple words.
In our x86\_64 implementation of verification,
we use AVX$2$ instructions to
further speed up the bitsliced $\FF_3$-vector arithmetic above
by moving from \(\beta = 64\) to \(\beta = 256\).
More precisely, we use
Intel\textsuperscript{\tiny\textregistered} intrinsics including
\textsc{\_mm256\_xor\_si256},
\textsc{\_mm256\_or\_si256},
and
\textsc{\_mm256\_and\_si256},
as well as special instructions to load and store
the $256$-bit words.

\subsection{Key Generation}

Algorithm~\ref{algo:keygen} requires a pseudorandom bit generator
to sample the matrices, vectors, and permutation.
In our implementation we use the SHAKE256 XOF,
seeded with a \(\lambda\)-bit random string
(which can be used as a compressed representation of the private key if
desired).
To generate the random permutation in \(\Sym{n}\),
we construct the list of integers from 0 to \(n-1\),
and then shuffle it using
the modern Fisher-Yates Shuffle algorithm~\cite{knuth97}.

We also used SHAKE256 to generate the randomness
required in our other algorithms---though for security,
each of these uses a different random string as initial state.

\subsection{Signing}

Apart from the use of \Hash (Algorithm~\ref{algo:hash}),
our signing implementation is heavily based on the
version of \Decode included in the proof-of-concept \Wave trapdoor code
from~\cite{Wave-implem}.
The implementation of~\cite{Wave-implem} generates trapdoor challenges,
using random trits in place of hash function output.
(Indeed, it neither specifies nor uses a hash function,
so it cannot be used to sign messages as is.)
Apart from some use of bitsliced arithmetic as in~\S\ref{sec:ternary}
for vector-matrix multiplication,
we made absolutely no attempt to optimize or improve the signing algorithm.

Indeed, we privileged correctness and safety (and coherence
with~\cite{2019/Debris-Alazard--Sendrier--Tillich})
over speed for signing:
any modifications to the original algorithm must be accompanied by
a subtle and difficult verification of the distributions produced by the
decoder.
A verified \emph{and} optimized signing algorithm
is an important future goal for \Wave development.

For completeness,
we include pseudocode for \Decode
as Algorithm \ref{algo:dec} in Appendix~\ref{ap:decoder}.
Our implementation uses precomputed tables
for the distribution of rejection with \SampleV and \SampleU,
but these values could be recomputed on the fly
using the parameters proposed in
Table~\ref{table:genParameters} and in Appendix~\ref{ap:parameters}.

To hash into \(\FF_3^{n-k}\),
our implementation uses Algorithm~\ref{algo:hash}
with SHA3-512 as \BinaryHash
and SHAKE-256 as the XOF in \Expand.

To sign with Algorithm~\ref{algo:sign},
we need to (re)generate the parity-check matrix.
While it is possible to avoid this by storing the parity check matrix
with the private key,
this imposes an unacceptably high space overhead
(this matrix is much bigger than the expanded private key,
which is already very large).

\begin{remark}
    We remind the reader that the signing process
    requires great care. One needs to respect the
    parameters and distributions involved in the rejection sampling to produce an output that
    is statistically indistinguishable from a random word of weight \(w\).
    Any deviation from this distribution could result in an attack on
    the signature scheme.
\end{remark}

\subsection{Verification}

In \Wave, the verification process
consist into check the weight of a vector, and compare the
hash of the message. However, we change this to simplify
the verification process.

In summary, Algorithm~\ref{algo:stverif-M}
consists of a for loop with additions and subtractions
of vectors. For those operations, we use the efficient
ternary vector arithmetic. After the loop, we
need to check the weight of $\vec{s}$ and $\vec{x}$.
To check the weight of a vector using Type 2 encoding,
we only need to count the number of ones in the low
part of the encoding, that is, given a
vector $\vec{a} = (\vec{a_l}, \vec{a_h})$ we
only need to compute $\BinaryWeight(\vec{a_l})$.

The Hamming weight, or \BinaryWeight function can be implemented
by counting the number of $1$'s in the binary representation
of an integer, or in the case the amount of $1$ bits in the
register. The binary weight for AVX2 can be computed
using~\cite{fastpop}. However, the generic code for
other architectures needs
to be done without AVX instructions. Listing~\ref{lst:pop}
shows the method for $32$-bit words; the technique extends easily
to $64$-bit words.

\begin{lstlisting}[language=C,style=customc, caption=BinaryWeight for $32$ bit architetures.,label={lst:pop}]
static uint32_t BinaryWeight (uint32_t n)
  {
      uint32_t m1 = (~(wave_word)0) / 3;
      uint32_t m2 = (~(wave_word)0) / 5;
      uint32_t m4 = (~(wave_word)0) / 17;

      n -= (n >> 1) & m1;
      n = (n & m2) + ((n >> 2) & m2);
      n = (n + (n >> 4)) & m4;
      n += n >>  8;
      n += n >> 16;
      return n & 0x7F;
  }
\end{lstlisting}

\subsection{Cortex-M4 implementation}

Implementing \Wavelet verification
for an ARM Cortex-M4 platform
presents several interesting challenges.

\paragraph{QSPI and Flash.}

\Wavelet has relatively compact signatures,
but large public keys---generally far too large
for the RAM available in many small devices.
In our case,
the nRF52840-dev Kit
provides $256$~\si{\kilo\byte} RAM.
However, it also has
$64$~\si{\mega\byte} of external QSPI (Queued Serial Peripheral
Interface) flash memory,
which allows reads of up to $16$ MB/sec
using the \texttt{nrf\_drv\_qspi\_read} instruction.\footnote{%
    Function documented at
    \url{https://infocenter.nordicsemi.com/index.jsp?topic=/com.nordic.infocenter.sdk5.v14.2.0/group__nrf__drv__qspi.html}.}
For this board, therefore,
we can flash the public key into the external flash
and read it later;
this is a useful approach for applications
where a device must verify signatures
from a signer known at installation time,
e.g. for secure software updates.

While this allows us to get the entire \Wavelet public key onto the
device, we still cannot load the entire matrix into SRAM.
In our Cortex-M4 implementation of Algorithm~\ref{algo:stverif-M}
we load one row at a time,
as and when it is used in a computation.
To fetch the data we have a base address, from which we can easily
compute the address for each required row.

\paragraph{Streaming vs Flashing public keys.}
Another way to handle big keys is to stream them;
this approach has been used for other postquantum signatures on
Cortex-M4~\cite{GonzalezHKKLSWW21}.
However,
as noted in~\cite{ChenC21}, this design choice depends on the application.
In our case, we have two reasons to keep the public key in
the flash. First, streaming the data is slower:
streaming speeds in~\cite{GonzalezHKKLSWW21}
were around $500$ kbits/s,
which is slower than QSPI.
Second, streamed keys require validation,
which would complicate and slow verification.

\section{
    Experimental results
}
\label{sec:results}
In our tests, we ran the code
on an Intel\textsuperscript{\tiny\textregistered} Core\textsuperscript{TM} i$7$-$10610$U processor
under Arch Linux (Kernel 5.14.11),
compiling with GCC version $11.1.0$.
We counted cycles using \texttt{cpucycles.h}
from SUPERCOP\footnote{\url{https://bench.cr.yp.to/supercop.html}},
disabling Turbo Boost and SpeedStep.
We first ran our tests $100$ times to clean
the cache, then $100$ more times to get
the average number of cycles.

Table~\ref{tab:intel} shows
the results of our implementation for KeyGen, Sign,
and \Wavelet Verification.
We also include timings for
an implementation of classic \Wave verification,
based on the same underlying arithmetic,
for comparison.
We implemented both verifications
twice:
once with AVX instructions and one without them.
(KeyGen and Sign do not use AVX instructions,
since we focused on optimizing verification.)

We also tested verification in a Nordic nRF52840 Development Kit,
which provides an ARM Cortex-M4 microcontroller running at \SI{64}{\mega\hertz},
with $256$ \si{\kilo\byte} RAM, $1$ \si{\mega\byte} flash and $64$ \si{\mega\byte} of
external memory. We compile the code using
GNU Arm Embedded Toolchain version $11.2.0$. We used the timer library provided
by Nordic Semiconductors to measure running time in ms and in ``ticks''.
Table~\ref{tab:cortex} presents these results.

Timings for signature compression and decompression
appear in Table~\ref{tab:encode}
(Intel\textsuperscript{\tiny\textregistered} Core\textsuperscript{TM})
and Table~\ref{tab:encode_arm} (ARM Cortex-M4).
Timings were measured in the same way as above.
We see that (de)compression times are negligible
in comparison with the other operations.

\begin{table}[!ht]
    \caption{\Wave and \Wavelet signature scheme timings on an Intel\textsuperscript{\tiny\textregistered} Core\textsuperscript{TM}.}
    \label{tab:intel}
    \centering
    \begin{tabular}{r|rrr}
        \toprule
        Operation & Optimization & Time (\SI{}{\si{\ms}})  & Cycles      \\
        \midrule
        Key Generation & \texttt{-O3} & $3144.00$ & $7\,403\,069\,461$ \\
        \midrule
        Sign & \texttt{-O3} & $702.77$ & $1\,644\,281\,062$ \\
        \midrule
        \multirow{2}{*}{\Wave Verification} & \texttt{-O3} & $2.20$ & $5\,098\,394$ \\
                                  & AVX instr & $2.19$ & $5\,064\,048$ \\

        \multirow{2}{*}{\Wavelet Verification} & \texttt{-O3} & $1.19$ & $2\,733\,457$ \\
                                    & AVX instr & $0.45$ & $1\,087\,538$ \\

        \bottomrule
    \end{tabular}
\end{table}

\begin{table}[!ht]
    \caption{\Wavelet verification time on an ARM Cortex-M4 (nRF52840-DK).}
    \label{tab:cortex}
    \centering
    \begin{tabular}{c|rrr}
        \toprule
        Operation & Optimization & Time (\SI{}{\si{\ms}})  & Ticks      \\
        \midrule
        \Wavelet Verification\_static & \texttt{-O2s} & $402$ & $13172$ \\
        \bottomrule
    \end{tabular}
\end{table}

\begin{table}[!ht]
    \caption{\Wavelet signature compression and decompression timings
    on an Intel\textsuperscript{\tiny\textregistered} Core\textsuperscript{TM}.}
    \label{tab:encode}
    \centering
    \begin{tabular}{rr|rrr}
        \toprule
        Encoding & Operation & Optimization & Time (\si{\ms})  & Cycles      \\
        \midrule
        \multirow{2}{*}{\DuoHuff} & Compression & \texttt{-O3} & $0.041$ & $79\,529$ \\
        & Decompression & \texttt{-O3} & $0.043$ & $81\,543$ \\
        \midrule
        \multirow{2}{*}{\TriHuff} & Compression & \texttt{-O3} & $0.065$ & $126\,230$ \\
        & Decompression & \texttt{-O3} & $0.057$ & $109\,546$ \\
        \bottomrule
    \end{tabular}
\end{table}

\begin{table}[!ht]
    \caption{\Wavelet signature compression and decompression timings on an ARM Cortex-M4 (nRF52840-DK).}
    \label{tab:encode_arm}
    \centering
    \begin{tabular}{rr|rrr}
        \toprule
        Encoding & Operation & Optimization &  Time (\SI{}{\si{\ms}})  & Ticks
        \\
        \midrule
        \multirow{2}{*}{\DuoHuff} & Compression & \texttt{-O2s} & $3$ & $123$
        \\
        & Decompression & \texttt{-O2s} & $7$ & $231$
        \\
        \midrule
        \multirow{2}{*}{\TriHuff} & Compression & \texttt{-O2s} & $3$ & $105$
        \\
        & Decompression & \texttt{-O2s} & $5$ & $171$
        \\
        \bottomrule
    \end{tabular}
\end{table}

Our code is available from \url{https://github.com/waveletc/wavelet}.

\section{
    Conclusion
}
\label{sec:conclusion}

This work presents the first complete implementation of the \Wave
postquantum signature scheme.
Our variant, \Wavelet,
includes important optimizations---from the highest level to the
lowest---yielding much shorter signatures and a much faster verification
algorithm.

To interface this fundamentally ternary signature scheme with the binary
world, we defined a practical hash function
and a highly efficient compression algorithm for signatures,
taking advantage of their unusually high weight.
This yields even shorter signatures
at very little cost.
For signatures targeting the 128-bit classical security level
(and NIST Level~1 postquantum security),
the \emph{Supertubos} parameter set yields \Wavelet signatures
that fit in less than a kilobyte.
\Wavelet signatures are therefore the smallest code-based signatures at this
security level,
and are competitive with signatures across other post-quantum paradigms.

We have also demonstrated the feasibility
of \Wavelet verification on ARM Cortex-M4 microcontroller platforms.
While \Wavelet public keys are generally much bigger than the available SRAM,
our implementation exploits large external flash where it is available.
While memory access latency remains a major bottleneck,
\emph{Supertubos} \Wavelet signatures decompress and verify
in under half a second on the Nordic nRF52840 Development Kit board.
Our verification algorithm should also be compatible with streaming public keys,
similar to the approach of~\cite{GonzalezHKKLSWW21},
though we have not attempted to implement this.

\paragraph{Future Work.}
Our results highlight some interesting outstanding problems, 
both theoretical and practical.
On the theoretical side,
\Wave for \(q > 3\) needs more analysis.
The derivation of alternative \Wave parameter sets---targeting NIST Levels 3 and 5,
for example, or tuning key and signature sizes for better performance in
embedded and other environments---is also a priority.

On the practical side,
key generation and signing must be improved and optimized.
Another crucial problem
is to develop a signing algorithm that is secure against
side-channel attacks, or at least constant-time.
For the current signing algorithm,
with its complexity
and dependence on rejection sampling,
this is a highly nontrivial problem.

\bibliographystyle{alpha}
\bibliography{references}

\newpage
\appendix
\section{The \Wave decoder}
\label{ap:decoder}

Algorithm~\ref{algo:dec} is the \Wave decoder,
which finds the error vector $\vec{e}$ with
Hamming weight $w$ such that $\vec{e}\Hsec = \vec{s}$.
It uses the precomputed data and rejection sampling parameters from Appendix~\ref{ap:parameters}.
It uses \ElimSingle (Algorithm~\ref{algo:ES})
from key generation,
and auxiliary functions
\FreeSetV (Algorithm~\ref{algo:FreeSetV})
and
\FreeSetU (Algorithm~\ref{algo:FreeSetU}).

\begin{algorithm}[!ht]
    \caption{The \Wave decoder}
    \label{algo:dec}
    \Function{\Decode{%
        $\vec{s}$,
        $\vec{R}_U$,
        $\vec{R}_V$,
        $\vec{a}$,
        $\vec{b}$,
        $\vec{c}$,
        $\vec{d}$}
    }{
        \KwIn{%
            $\vec{s}\in\F_3^{n-k}$,
            $\vec{R}_U \in \F_3^{(n/2-k_U)\times k_U}$,
            $\vec{R}_V \in \F_3^{(n/2-k_V)\times k_V}$,
            $\vec{a}$,
            $\vec{b}$,
            $\vec{c}$,
            $\vec{d} \in \F_3^{n/2}$
        }
        \KwOut{%
            $\vec{e}\in\F_{3}^n$
            such that $|\vec{e}|=w$
            and $\vec{e}\transpose{\Hsec}=\vec{s}$
            \\
            \qquad
            \qquad
            where $\Hsec=$
            \ParityCheckUV{$\vec{R}_U$, $\vec{R}_V$, $\vec{a}$, $\vec{b}$, $\vec{c}$, $\vec{d}$}
        }
        $\vec{s}_{U} \gets (s_0,\ldots,s_{n/2-k_U-1})$
        \tcp*{\(\vec{s}_{U} \in \F_{3}^{n/2-k_U}\)}
        $\vec{s}_{V} \gets (s_{n/2-k_U},\ldots,s_{n-k-1})$
        \tcp*{\(\vec{s}_{V} \in \F_{3}^{n/2-k_V}\)}
        $\vec{H}\gets(\mathbf{I}_{n/2-k_V}\mid\vec{R}_V)$
        \tcp*{\(\vec{H} \in \F_{3}^{(n/2-k_V)\times n/2}\)}
        \Repeat(\tcp*[f]{rejection sampling}){\AcceptSampleV{$t$}}{
            $\ell \gets $ \SampleV{}
            \;
            $\vec{x}_{V} \Unif \F_{3}^{d}$
            \;
            $\vec{x}_{V}' \Unif \big\{ \vec{x} \in \F_{3}^{k_{V}-d} \colon |\vec{x}| = \ell \big\}$
            \;
            $\left( \left( \vec{I}_{n/2-k_{V}} | \vec{A} \right),\vec{s}_{0}, \pi_{V}\right)\gets$  \FreeSetV{$\vec{H},\vec{s}_{V}$}\;
            $\vec{e}_{V} \gets \pi_{V}^{-1}\left( \left(\vec{s}_{0} - \left( \vec{x}_{V} \parallel \vec{x}'_{V} \right)\transpose{\vec{A}} \right)  \parallel \left( \vec{x}_{V} \parallel \vec{x}'_{V} \right) \right)$
            \;
            $t \gets |\vec{e}_{V}|$\;
        }
        $\vec{H}\gets(\mathbf{I}_{n/2-k_U}\mid\vec{R}_U)$
        \tcp*{\(\vec{H} \in \F_{3}^{(n/2-k_U)\times n/2}\)}
        $\mathcal{E}_{V} \gets$ \Supp{$\vec{e}_{V}$}
        \tcp*{\Supp{$\vec{x}$} = list of coordinates where $x_{i} \neq 0$}
     	$\vec{a}_{V} \gets (\vec{b}\hsp\vec{e}_V\parallel\vec{d}\hsp\vec{e}_V)$ \tcp*{for $\vec{x},\vec{y}\in \F_{3}^{n/2}$, $\vec{x}\hsp \vec{y} \eqdef (x_{i}y_{i})_{0\leq i \leq n/2-1}$}
        \Repeat(\tcp*[f]{rejection sampling}){%
            \AcceptSampleU{$m_{1}(\vec{e})$, $t$} \Where $m_{1}(\vec{x})
            \eqdef \#\{0 \leq i \leq \frac{n}{2}-1 : x_{i} \neq x_{\frac{n}{2}+i}\}$
        }{
            $k_{\neq 0} \gets$ \SampleU{$t$}
            \;
            $\left( \left(\vec{I}_{n/2-k_{U}} | \vec{A} \right),\vec{s}_{1}, \pi_{U}\right) \gets$ \FreeSetU{$k_{\neq 0}$, $\mathcal{E}_{V}$, $\vec{H}$, $\vec{s}_{U}$}
            \;
            $(z_{i})_{0 \leq i \leq n/2-1} \gets \pi_{U}(\vec{e}_V)$
            \;
            \Repeat{$|\vec{e}| = w$}{
                $\vec{x}_{U} \Unif \F_{3}^{d}$
                \;
                $\vec{x}_{U}' \Unif \big(\F_{3}\backslash \{0\}\big)^{k_{U} - k_{\neq 0} - d}$
                \;
                \(
                    \vec{x}_{U}''
                    \gets
                    ( x_{j})_{0 \leq j < k_{\neq 0}}
                \)
                \Where
                \(
                    x_{j} \gets
                    \frac{b_u}{a_u}z_{n/2 - k_{\neq 0} + j}
                    +\frac{d_u}{c_u}z_{n/2 - k_{\neq 0} + j} \mbox{ where } u \eqdef {\pi_{U}^{-1}(j)}
                \)
                \;
                $\vec{e}_{U} \gets \pi_{U}^{-1}\left( \left(\vec{s}_{1} - \left( \vec{x}_{U}\parallel \vec{x}_{U}'\parallel\vec{x}''_{U} \right)\transpose{\vec{A}}\right) \parallel \left( \vec{x}_{U}\parallel \vec{x}_{U}'\parallel\vec{x}''_{U} \right)\right)$
                \;
                $\vec{a}_{U} \gets (\vec{a}\hsp\vec{e}_U\parallel\vec{c}\hsp\vec{e}_U)$
                \;
                $\vec{e} \gets \vec{a}_{V} + \vec{a}_{U}$
                \;
            }
        }
        \Return{$\vec{e}$}
    }
\end{algorithm}

The \FreeSetV and \FreeSetU algorithms are superficially similar,
but there are important differences in their
input and permutation sampling.
In particular,
\FreeSetU uses an auxiliary function
\RandPermSet{$n$, $k_{\neq 0}$, $(i_{0},\dots,i_{t-1})$},
which samples a random permutation $\pi \in \mathcal{S}_{n}$
subject to the constraints
\begin{align*}
    \left\{ \pi(0),\dots,\pi(t-k_{\neq 0}-1) \right\}
    & \subseteq
    (i_{0},\dots,i_{t-1})
    \shortintertext{and}
    \left\{\pi(t-k_{\neq 0}),\dots,\pi(n-k_{\neq 0}-1) \right\}
    & =
    \{1,\dots,n\} \backslash \{i_{0},\dots,i_{t-1}\}
    \,.
\end{align*}

\begin{algorithm}[ht]
    \caption{FreeSetV}
    \label{algo:FreeSetV}
    \KwIn{$\vec{H}\in \F_3^{(n-k)\times n}$, $\vec{s}\in\F_{3}^{n-k}$}
    \KwOut{%
        $\vec{H} =  (\vec{I}_{n-k} | \vec{M}) \in \F_3^{(n-k)\times n}$
        where $\vec{M} \in \F_3^{(n-k)\times n}$,
        $\vec{s}\in\F_{3}^{n-k}$, $\pi \in\Sym{n}$
    }
    \Function{\FreeSetV{$\vec{H},\vec{s}$}}{
        \SetKwData{Pivot}{pivot}
        \SetKwData{Nonpivot}{nonpivot}
        \Repeat{$\#\Nonpivot \leq d$}{
            $\pi \Unif \Sym{n}$
            \;
            $\Pivot \gets ()$
            \;
            $\Nonpivot \gets ()$
            \;
            $(r, c) \gets (0, 0)$
            \;
            \While{$r < n-k$ \And $c < n$}{
                \((\vec{H}\parallel\transpose{\vec{s}},P) \gets\) \ElimSingle{$\vec{H}\parallel \transpose{\vec{s}}$, $r$, $\pi(c)$}
                \tcp*{Algorithm~\ref{algo:ES}}
                \uIf{$P$}{
                    $(\Pivot,r) \gets (\Pivot\parallel(\pi(c)), r+1)$
                }
                \Else{
                    $\Nonpivot \gets \Nonpivot\parallel(\pi(c))$
                }
                $c \gets c + 1$
                \;
            }
        }
        $\pi \gets \Pivot \parallel \Nonpivot \parallel (\pi(c),\dots,\pi(n-1))$ \;
        \Return{$(\pi(\vec{H}), \vec{s},\pi)$}
    }
\end{algorithm}

\begin{algorithm}[ht]
	\caption{FreeSetU}
    \label{algo:FreeSetU}
    \KwIn{$\mathcal{E} \subseteq \{0,\ldots,n-1\}$, $k_{\neq 0}\in\{0,\dots\#\mathcal{E}\}$, $\vec{H}\in \F_3^{(n-k)\times n}$, $\vec{s}\in\F_{3}^{n-k}$}
    \KwOut{%
        $\vec{H} =  (\vec{I}_{n-k} | \vec{M}) \in \F_3^{(n-k)\times n}$
        where $\vec{M} \in \F_3^{(n-k)\times n}$,
        $\vec{s}\in\F_{3}^{n-k}$, $\pi \in\Sym{n}$
    }
    \Function{\FreeSetU{$k_{\neq 0}$, $\mathcal{E}$, $\vec{H},\vec{s}$}}{
        \SetKwData{Pivot}{pivot}
        \SetKwData{Nonpivot}{nonpivot}
        \Repeat{$\#\Nonpivot \leq d$}{
            $\pi \gets$ \RandPerm{$n$, $k_{\neq 0}$, $\mathcal{E}$}
            \;
            $\Pivot \gets ()$
            \;
            $\Nonpivot \gets ()$
            \;
            $(r, c) \gets (0, 0)$
            \While{$r < n-k$ \And $c < n$}{
                \(
                    (\vec{H}\parallel \transpose{\vec{s}},P) 
                    \gets
                    \ElimSingle(\vec{H}\parallel \transpose{\vec{s}}, r, \pi(c))
                \)
                \tcp*{Algorithm~\ref{algo:ES}}
                \uIf{$P$}{
                    $(\Pivot,r) \gets (\Pivot\parallel(\pi(c)), r+1)$
                }
                \Else{
                    $\Nonpivot \gets \Nonpivot\parallel(\pi(c))$
                }
                $c \gets c + 1$
                \;
            }
        }
        $\pi \gets \Pivot \parallel \Nonpivot \parallel (\pi(c),\dots,\pi(n-1))$
        \;
        \Return{$(\pi(\vec{H}),\vec{s},\pi)$}
    }
\end{algorithm}

\section{Rejection sampling parameters}
\label{ap:parameters}

There are two rejection sampling steps in the signing algorithms of \Wave
(specifically, in the \Wave decoder).
We split the description of their parameters in two parts:
the \emph{$V$-part} and the \emph{$U$-part}.
We write
$$
k_{V}' \eqdef k_{V} -d \quad \mbox{and} \quad k_{U}' \eqdef k_{U} - d
$$
where $k_{U},k_{V}$ and $d$ are given in Table \ref{table:genParameters}.

\paragraph{The $V$-part.}
This first rejection sampling step involves several functions
on \(\{0,\ldots,n/2\}\), namely
\begin{align*}
    \label{eq:fvrs}
    f_{V}^{\textup{rs}}(i)
    & \eqdef
    \frac{1}{\mathop{\max}\limits_{\substack{0 \leq j \leq n/2}}\frac{q_{1}^{\textup{unif}}(j)}{q_1(j)} \quad}\; \frac{q_{1}^{\textup{unif}}(i)}{q_{1}(i)}
    \,,
    \\
    q_{1}^{\textup{unif}}(i)
    & \eqdef
    \frac{\binom{n/2}{i}}{\binom{n}{w}2^{w/2}} \mathop{\sum}\limits_{\substack{p=0 \\
			w+p \equiv 0 \mod 2}}^{i}\binom{i}{p}\binom{n/2-i}{(w+p)/2-i}2^{3p/2}
    \,,
    \shortintertext{and}
    q_1(i)
    & \eqdef
    \mathop{\sum}\limits_{t=0}^{i} \frac{\binom{n/2-k_V'}{i-t}2^{i-t}}{3^{n/2-k_V'}} p_{V}(t)
    \,,
\end{align*}
where
the function $p_{V}(\cdot)$ is a system parameter for the $V$-part of the rejection sampling.

\paragraph{The $U$-part.}
This second rejection sampling step involves several functions on the
domain
\[
    \left\{
        (s,t) :
        t_\textup{min} \le t \le t_\textup{max},
        0 \le s \le \min(t,n-w),
        s \equiv w \pmod{2}
    \right\}
\]
where
\[
    t_{\textup{min}} \eqdef 1953
    \qquad
    \text{and}
    \qquad
    t_{\textup{max}} \eqdef 2745
    \,;
\]
these functions are
\begin{equation}\label{eq:furs}
    f_{U}^{\textup{rs}}(s,t) \eqdef \frac{1}{\mathop{\max}\limits_{\substack{0 \leq u \leq v}}\frac{q_{2}^{\textup{unif}}(u,v)}{q_2(u,v)} \quad}\; \frac{q_{2}^{\textup{unif}}(s,t)}{q_{2}(s,t)}
\end{equation}
where
\[
    q_{2}^{\textup{unif}}(s,t)  \eqdef  \frac{\binom{t}{s}\binom{n/2 - t}{\frac{w+s}{2}-t}2^{\frac{3s}{2}}}{\sum\limits_{p} \binom{t}{p}\binom{n/2-t}{\frac{w+p}{2}-t}2^{\frac{3p}{2}}}
\]
and
\[
    q_2(s,t)  \eqdef \mathop{\sum}\limits_{\substack{t + k_U' - n/2 \leq k_{\neq 0} \leq t \\ k_0 = k_U' - k_{\neq 0} }} \frac{\binom{t - k_{\neq 0}}{s}\binom{n/2 - t - k_0}{\frac{w+s}{2} - t - k_0}2^{\frac{3s}{2}}}{\mathop{\sum}\limits_{p} \binom{t - k_{\neq 0}}{p}\binom{n/2 - t - k_0}{\frac{w+p}{2} - t - k_0}2^{\frac{3p}{2}} } \; p_{U}(k_{\neq 0},t)
\]
where the function $p_{U}(\cdot,t)$
is a system parameter for the $U$-part.

\paragraph{Precomputed Data.}
The values for
$(p_{V}(t))_{1 \leq i \leq n/2}$
and
$(p_{U}(\cdot,t))_{t_{\textup{min}} \leq t \leq t_{\textup{max}}}$
used in our implementation
can be found at \url{https://github.com/waveletc/wavelet} 
in the file \texttt{precomputedData.txt}.
We precomputed and stored the necessary values
of $f_{V}^{\textup{rs}}$
and $(f_{U}^{\textup{rs}}(\cdot,t))_{t_{\textup{min}} \leq t \leq t_{\textup{max}}}$
to 128-bit precision,
to save recomputing $f_{V}^{\textup{rs}}$
and $f_{U}^{\textup{rs}}$ each time they are called.
This requires significant storage;
if this is not available,
then the values can always be computed on the fly
using the formulas above.

\section{Signature compression and decompression algorithms}
\label{ap:comp}

For efficient signature compression and decompression
using \DuoHuff and \TriHuff,
we follow the algorithms of Moffat and
Turpin~\cite{1997/Moffat--Turpin}.
Tables~\ref{tab:lists_duo} and~\ref{tab:lists_tri} 
give the precomputed arrays required by the decompression algorithm.
The notation $[x] * n$ means a list consisting of the number $x$,
repeated $n$ times.

\begin{table}[!ht]
  \caption{Pre-computed lists for fast \DuoHuff decoding using the
    algorithm of~\cite{1997/Moffat--Turpin}.}
  \label{tab:lists_duo}
  \centering
  \begin{tabular}{r|l}
  \toprule
      \texttt{Code\_len} & $[2, 2, 3, 3, 4, 4, 4, 5, 5]$\\
  \midrule
      \texttt{First\_symbol} & $[0, 0, 0, 2, 4, 7, 10]$ \\
  \midrule
      \texttt{First\_code\_r} & $[0, 0, 0, 4, 12, 30, 32]$  \\
  \midrule
      \texttt{First\_code\_l} & $[0, 0, 0, 16, 24, 30, 32]$ \\
  \midrule
      \texttt{Search\_start} &  $[2] * 16 \parallel [3] * 8 \parallel [4] * 5 \parallel [5] * 3$\\
  \bottomrule
  \end{tabular}
\end{table}

\begin{table}[!ht]
  \caption{Pre-computed lists for fast \TriHuff decoding using the algorithm
    of~\cite{1997/Moffat--Turpin}.}
  \label{tab:lists_tri}
  \centering
  \begin{tabular}{r|l}
  \toprule
      \texttt{Code\_len} & $[3]*5 \parallel [4]*3 \parallel [6]*11 \parallel [7,9] \parallel [10]*6$\\
  \midrule
      \texttt{First\_symbol} & $[0, 0, 0, 0, 5, 8, 8, 19, 20, 20, 21, 27]$ \\
  \midrule
      \texttt{First\_code\_r} & $[0, 0, 0, 0, 10, 10, 52, 126, 126, 508, 1018, 1024]$  \\
  \midrule
      \texttt{First\_code\_l} & $[0, 0, 0, 0, 640, 832, 832, 1008, 1016, 1016, 1018, 1024]$ \\
  \midrule
      \texttt{Search\_start} &  $[3] * 640 \parallel  [4] * 192 \parallel [6] * 176 \parallel [7] * 8 \parallel [9] * 2 \parallel [10] * 6$\\
  \bottomrule
  \end{tabular}
\end{table}

\section{
    Mathematical proofs
}
\label{appendix:proofs}

\paragraph{Proof of Proposition~\ref{prop:s-hat-weight}.}
Recall that for each \(\vec{s}\) in \(\FF_3^k\),
the vector \(\vec{\hat{s}}\) in \(\FF_3^k\)
is defined by
\[
    \begin{rcases}
        \hat{s}_{2i} := s_{2i} + s_{2i+1}
        \\
        \hat{s}_{2i+1} := s_{2i} - s_{2i+1}
    \end{rcases}
    \quad
    \text{for}
    \quad
    0 \le i < k/2
    \,,
    \quad
    \text{and}
    \quad
    \hat{s}_{k-1} := s_{k-1}
    \text{ if \(k\) is odd}
    \,.
\]
We want to show that
\[
    |\vec{\hat{s}}| + |\vec{s}|
    =
    \frac{3}{2}(k+\epsilon) - 3\delta
    \,,
\]
where
\[
    \delta = \#\{0 \le i < k/2 \mid s_{2i} = s_{2i+1} = 0\}
\]
and
\[
    \epsilon 
    =
    \begin{cases}
        0 & \text{if \(k\) is even}\,,
        \\
        1 & \text{if \(k\) is odd and \(s_{k-1} \not= 1\)}\,,
        \\
        -1 & \text{if \(k\) is odd and \(s_{k-1} = 0\)}\,.
    \end{cases}
\]

First, suppose \(k\) is even.
Breaking \(\vec{\hat{s}}\) into pairs \((\hat{s}_{2i},\hat{s}_{2i+1})\)
for \(0 \le i < k/2\), we have
\begin{itemize}
    \item
        \(|(\hat{s}_{2i},\hat{s}_{2i+1})| = 0 \iff |(s_{2i},s_{2i+1})| = 0\),
        and there are \(\delta\) such pairs;
    \item
        \(|(\hat{s}_{2i},\hat{s}_{2i+1})| = 2 \iff |(s_{2i},s_{2i+1})| = 1\),
        and there are \(k-|\vec{s}| - 2\delta\) such pairs;
    \item
        \(|(\hat{s}_{2i},\hat{s}_{2i+1})| = 1 \iff |(s_{2i},s_{2i+1})| = 2\),
        and there are 
        \(|\vec{s}| - k/2 + \delta\) such pairs.
\end{itemize}
Summing over the pairs, we find
\[
    |\vec{\hat{s}}| + |\vec{s}|
    = 
    (0+0)\delta 
    + (2+1)(k-|\vec{s}| - 2\delta) 
    + (1+2)(|\vec{s}| - k/2 + \delta)
    =
    3k/2 - 3\delta 
    \,.
\]

If \(k\) is odd and \(s_{k-1} = 0\)
then \(\hat{s}_{k-1} = s_{k-1} \not= 0\),
so
\[
    |\vec{\hat{s}}|+|\vec{s}|
    = 
    |(\hat{s}_0,\ldots,\hat{s}_{k-2})| + |(s_0,\ldots,s_{k-2})|
    \,,
\]
and applying the even-length argument above to the subvectors
yields 
\[
    |\vec{\hat{s}}|+|\vec{s}| = 3(k-1)/2 - 3\delta
    \,.
\]

Similarly, if \(k\) is odd and \(s_{k-1} \not= 0\),
then \(\hat{s}_{k-1} = s_{k-1} \not= 0\),
so 
\[
    |\vec{\hat{s}}|+|\vec{s}|
    = 
    2 + |(\hat{s}_0,\ldots,\hat{s}_{k-2})| + |(s_0,\ldots,s_{k-2})|
    \,,
\]
and applying the even-length argument again yields 
\[
    |\vec{\hat{s}}|+|\vec{s}| = 2 + 3(k-1)/2 - 3\delta = 3(k+1)/2 - 3\delta
    \,.
\]

\paragraph{Proof of Proposition~\ref{prop:M}.}
Let \(\vec{R}\) be in \(\FF_3^{(n-k)\times k}\).
We want to show that for every \(\vec{s}\) in \(\FF_3^k\),
if \(\vec{\hat{s}}\) is defined as above,
then
\[
    \vec{s}\transpose{\vec{R}}
    = 
    -\vec{\hat{s}}\vec{M} 
\]
where \(\vec{M}\) is the matrix in \(\FF_3^{k\times(n-k)}\)
whose rows are
\[
    \begin{rcases}
        \vec{M}_{2i} := \transpose{\vec{R}}_{2i} + \transpose{\vec{R}}_{2i+1}
        \\
        \vec{M}_{2i+1} := \transpose{\vec{R}}_{2i} - \transpose{\vec{R}}_{2i+1}
    \end{rcases}
    \quad
    \text{for}
    \quad
    0 \le i < k/2
    \,,
    \quad
    \text{and}
    \quad
    \vec{M}_{k-1} := -\transpose{\vec{R}}_{k-1}
    \text{ if \(k\) is odd}
\]
(note the minus sign).
To show this, we simply sum the identities
\begin{align*}
    \hat{s}_{2i}\vec{M}_{2i} + \hat{s}_{2i+1}\vec{M}_{2i+1} 
    & =
    -(s_{2i}\transpose{\vec{R}}_{2i} + s_{2i+1}\transpose{\vec{R}}_{2i+1})
    \,,
    \intertext{for \(0 \le i < k/2\), and add}
    \hat{s}_{k-1}\vec{M}_{k-1} 
    & = 
    -s_{k-1}\transpose{\vec{R}}_{k-1}
\end{align*}
if \(k\) is odd.

\end{document}